\documentclass[journal,twocolumn]{IEEEtran}
\IEEEoverridecommandlockouts
\usepackage{cite}
\usepackage{amsmath,amssymb,amsfonts}
\usepackage{graphicx}
\usepackage{epstopdf}
\usepackage{subfigure}
\usepackage{stfloats}
\usepackage{textcomp}
\usepackage{xcolor}
\usepackage{flushend}
\usepackage{array}
\usepackage{algorithm}
\usepackage{mathtools}
\usepackage{algpseudocode}
\usepackage{multirow}
\usepackage{amsthm}

\newtheorem{corollary}{Corollary}
\newtheorem{lemma}{Lemma}
\newtheorem{proposition}{Proposition}
\newtheorem{remark}{Remark}

\def\BibTeX{{\rm B\kern-.05em{\sc i\kern-.025em b}\kern-.08em
    T\kern-.1667em\lower.7ex\hbox{E}\kern-.125emX}}

\DeclareMathOperator{\sinc}{sinc}

\algnewcommand\algorithmicforeach{\textbf{for each}}
\algdef{S}[FOR]{ForEach}[1]{\algorithmicforeach\ #1\ \algorithmicdo}

\begin{document}

\title{On the Security of Directional Modulation via Time Modulated Arrays Using OFDM Waveforms}

\author{
 Zhihao~Tao and~Athina~Petropulu
\thanks{The authors are with the Department of Electrical and Computer Engineering, Rutgers The State University of New Jersey, New Brunswick, NJ, USA (email: zhihao.tao@rutgers.edu; athinap@soe.rutgers.edu). This work was supported by ARO grant W911NF2320103 and NSF grant ECCS-2320568.}
}

\maketitle

\begin{abstract}
Time-modulated arrays (TMAs) transmitting information bearing orthogonal frequency division multiplexing (OFDM) signals can achieve directional modulation.  By turning its antennas on and off in a periodic fashion, the TMA can be configured to transmit the OFDM signal undistorted in the direction of a legitimate receiver and scrambled everywhere else. This capability  has been proposed as means of  securing the transmitted information from unauthorized users.
In this paper, we investigate how secure the TMA OFDM system is, by looking at the transmitted signal from an eavesdropper's point of view.
We demonstrate that the symbols observed by the eavesdropper across the OFDM subcarriers are linear combinations of the source symbols, with mixing coefficients that are unknown to the eavesdropper. We propose the use of independent component analysis (ICA) theory to obtain the mixing matrix and provide methods to resolve the  column permutation and scaling ambiguities, which are inherent in the ICA problem, by leveraging the structure of the mixing matrix and assuming knowledge of the characteristics of the TMA OFDM system.
In general, resolving the ambiguities and recovering the symbols requires long data. Specifically for the case of the constant modulus symbols, we propose a modified ICA approach, namely the constant-modulus ICA (CMICA), that provides a good estimate of the mixing matrix  using a small number of received samples.
%
We also propose countermeasures which the TMA could undertake in order to defend the scrambling. Simulation results are presented to demonstrate the effectiveness, efficiency and robustness of our scrambling defying and defending schemes.

\end{abstract}


\section{Introduction}

The broadcast nature of wireless transmission renders wireless and mobile communication systems vulnerable to eavesdropping. Physical layer security (PLS) approaches, originating from Wyner's wiretap channel work \cite{Wyner1975Wire}, offer information secrecy by exploiting the physical characteristics of the wireless channel. PLS methods can complement traditional cryptographic approaches, particularly in scenarios where the latter methods encounter difficulties in providing low latency and scalability due to challenges with key management or computational complexity \cite{poor2017wireless, qiu2023decomposed}.

Directional modulation (DM) \cite{daly2009dire} is a promising physical layer security (PLS) technique that has attracted significant interest in recent years. DM transmits digitally modulated signals intact only along pre-selected spatial directions, while distorting the signal in all other directions \cite{su2021secure, Xiao2023Synthesis}.  Compared with PLS approaches such as cooperative relaying strategies \cite{dong2009improving,Li2011cooperative,Li2020relay} and transmission of artificial noise \cite{zhang2019AN,wang2017AN}, DM-based methods are more energy- and cost-efficient \cite{su2022secure}.
%

DM can be implemented through waveform design or by modifying the transmitter hardware. For the former approach, \cite{daly2009dire} uses phase shifters to adjust the phases of each symbol, while \cite{Ottersten2016} and \cite{Alodeh2016DM} employ transmit precoders.
The works of \cite{su2022secure,khandaker2018secure,khandaker2018constructive} exploit constructive interference in designing transmit waveforms, where the alignment of the received signal with the intended symbols is not required, but rather, the signal is shifted away from the detection boundary of the signal constellation. The methods of \cite{daly2009dire,Ottersten2016,Alodeh2016DM, su2022secure,khandaker2018secure,khandaker2018constructive} require the location information on the eavesdroppers or channel state information (CSI)
which increases the communication overhead.
%
%
The works of \cite{valliappan2013antenna,ding2017Circular,Hamdi2016subset,ding2017free}  operate on the transmitter hardware and do not require CSI nor the location  of eavesdroppers. For example,  \cite{valliappan2013antenna} adopts a large antenna array working at millimeter-wave frequencies and proposes an antenna subset modulation-based DM technique. By appropriately selecting a subset of antennas for the transmission of each symbol, the radiation pattern can be modulated in a direction-dependent way, which yields  randomness to the constellations seen from directions other than the intended angles.
In \cite{ding2017free}, a retrodirective array, 
is proposed to implement the DM functionality.
Using a pilot signal provided by the legitimate receiver and appropriately designed weights, the retrodirective array creates a far-field radiation pattern consisting of two components: the information pattern and the interference pattern. The interference pattern is zero only along the direction of the legitimate user, thereby distorting the information signals in all other directions.
Time-modulated arrays (TMAs)  \cite{tvt2024security, manica2009almost, Massa20144d,Massa2018time} is another DM approach that operates on the transmit hardware but also  introduces time as an additional degree of freedom in the DM design. 
TMAs use switches to periodically
connect and disconnect the transmit antennas to the RF chain \cite{Kummer1963new, Kummer1963Ultra}. In \cite{manica2009almost, Massa20144d,Massa2018time}, which consider the single carrier system and transmit one symbol at a time, the radiation pattern of the array in each symbol is optimally computed via  global optimization tools, e.g., evolutionary algorithms, so that the transmitted signals are delivered  undistorted within a desired angular region,  while they are maximally distorted  elsewhere. Even though TMA-based approaches are more flexible as compared to other DM methods, they \cite{manica2009almost, Massa20144d,Massa2018time} 
involve computationally intensive optimization methods, and thus their complexity increases in dynamic environments, where the system configuration needs to change.
%
A low computational complexity  TMA DM approach has been proposed via the use of 
orthogonal frequency-division multiplexing (OFDM) transmit waveforms \cite{tvt2019time}.
By appropriately selecting the TMA parameters, the transmitter sends a scrambled signal in all directions except toward the legitimate destination. The scrambling effect occurs because the designed periodic antenna activations generate harmonics at the OFDM subcarrier frequencies, causing symbols on each subcarrier to mix with symbols from other subcarriers. The TMA parameters can be derived using closed-form expressions and simple rules, allowing the DM functionality to be implemented by configuring the transmitter hardware according to these rules, without requiring global optimization. As a result, the OFDM TMA offers low complexity and is easy to deploy in dynamic scenarios.
The TMA transmitter described in \cite{tvt2019time} is also applicable to modern wireless communication systems that support multiple carriers. These advantages make OFDM TMAs highly attractive for achieving directional modulation.

DM via TMAs transmitting OFDM waveforms has been studied in various applications, e.g., antenna array designs \cite{Wu2022Metamaterial, vosoughitabar2023directional, nooraiepour2023programming}, multicarrier systems \cite{shan2021multi}, target localization \cite{shan2022target}, joint communication and sensing systems \cite{zhaoyi2022TMA} and intelligent reflecting surface systems \cite{xu2024tma}, where their good potential for enhancing PLS has been demonstrated. {However, existing studies have primarily focused on TMA hardware implementation, energy efficiency improvement, and ON-OFF pattern design, while largely overlooking the security of the TMA OFDM system. The study in \cite{nooraiepour2022time} concludes that DM can effectively prevent such spoofing. Also, in  \cite{li2022chaotic}, the authors argue that the DM via TMAs transmitting OFDM waveforms has weak security due to the limited randomness of the periodic time modulation pattern and propose a chaotic-enabled phase modulation for TMA to enhance wireless security. However,  \cite{li2022chaotic} does not explore the possibility of an eavesdropper defying the TMA security.}

In this paper, we investigate the level of security provided by TMA-induced scrambling and demonstrate, for the first time, that the TMA OFDM system is not sufficiently secure unless specific measures are taken.
Specifically, we first show that the vector of the symbols received by the eavesdropper on all OFDM subcarriers 
can be expressed as the product of a mixing matrix and a vector containing the information symbols. The mixing matrix has a Toeplitz structure, but is  otherwise unknown to the eavesdropper as it depends on the TMA parameters.
{We show that, based on the received scrambled data, the eavesdropper can obtain the mixing matrix via an ICA-based approach, and also resolve all  ambiguities.}
 For the case of  constant modulus symbols,
  we propose an  ICA variant, referred to here as  constant-modulus ICA (CMICA).
{Assuming that the data symbols are non-Gaussian,  ICA recovers the symbols by designing an unmixing matrix, which, when applied to the received data results in non-Gaussian data.}
 The ICA objective  is a measure of 
 non-Gaussianity that needs to be maximized.
 %
%
 Initially,  CMICA is identical to ICA \cite{bingham2000fast} with a  Newton iteration,  in each iteration decorrelating the  unmixing matrix  in order to recover  independent  source signals.
 When using a small sample size, there may still be some residual dependence among the recovered signals even after convergence of the initial iteration. To address this, CMICA introduces a fine-tuning stage, where the iteration continues without the decorrelation step. Instead of continuing with the Newton iteration, CMICA switches to the gradient descent method, allowing for adjustable step sizes and better control during the fine-tuning process. The result from the Newton iteration serves as the initialization for the gradient descent iteration. Extensive simulations demonstrate that CMICA improves the estimation of the mixing matrix, even with limited data samples.

Regarding resolving the ICA ambiguities, we propose  a novel  $k$-nearest neighbors (KNN)-based  approach  that leverages the mixing matrix Toeplitz structure.
In particular, we first construct a similarity measure that allows us to rank matrices based on their resemblance to a Toeplitz matrix. Using this similarity measure, along with knowledge about the TMA OFDM system—such as the data constellation and the rules for selecting TMA parameters—we demonstrate how to identify the best matrix from among the column-reordered versions of the estimated mixing matrix.
Finally, we identify two scenarios in which the proposed defying scheme would fail, i.e.,  when the mixing matrix is rank-deficient,  or when there  TMA ON-OFF switching pattern is not unique. Based on these situations, we design defensive mechanisms that the transmitter can use to protect the scrambling process against eavesdroppers.

The novel contributions of this paper are summarized as follows:
\begin{enumerate}
\item We show that the recovery of the transmitted symbols based on the received scrambled symbols can be addressed as an ICA problem. We propose a low-complexity approach, which, under certain assumptions can  
  solve the ICA problem and all ambiguities involved,  thereby circumventing the TMA scrambling.
\item We propose an ICA variant called CMICA,  which is particularly well suited for  constant modulus symbols and works well for case of short observation lengths.

%
To eliminate ICA-inherent scaling and permutation ambiguities, we propose  a novel  $k$-nearest neighbors (KNN)-based  approach  that leverages the mixing matrix Toeplitz structure. We first construct  a measure that quantifies  the
Toeplitz resemblance of a matrix, and then 
 along with knowledge about the TMA OFDM system—such as the data constellation and the rules, we select the best fit from among the column-reordered versions of the
estimated mixing matrix. 

\item We identify two cases in which the scrambling is strong and thus secure transmission can be guaranteed. {The first scenario occurs when the mixing matrix is rank-deficient, leading to multiple possible solutions for the transmitted symbols. The second scenario arises when there is non-uniqueness in the ON-OFF switching pattern of the TMA, meaning that the scrambled signals can correspond to multiple configurations of the TMA parameters.}
{Furthermore, we propose two scrambling defense mechanisms. The first mechanism involves rotating the TMA transmitter to create a specific angular offset between the legitimate receiver and the eavesdropper, ensuring that the conditions for secure transmission are met. The second mechanism introduces controlled variations to the ON-OFF switching pattern or the transmitted symbols, complicating the eavesdropper's ability to exploit the statistical structure of the signals necessary for ICA-based attacks.} These cases and mechanisms reveal strategies for enhancing the wireless security of TMA in the future. 


\end{enumerate}

%

{Preliminary results of this work are presented in \cite{tao2024tma} and \cite{tao2023tma}. Compared to those publications, here, we  further improve the efficiency of the ICA-based estimation method and the robustness of the ambiguity resolving algorithm.
We also provide deeper insights into the wireless security of the TMA OFDM-enabled DM system and propose methods to enhance its security. To the best of our knowledge, this paper is the first  to assess and analyze the wireless security of directional modulation via TMA OFDM.}

The remainder of this paper is organized as follows. In Section II, we describe the system model of the TMA OFDM-enabled DM transmitter. In Section III, we elaborate the proposed defying scheme, including the CMICA-based mixing matrix estimation approach and the KNN-based ambiguity resolving algorithm. In Section IV, we illuminate how to defend the defying of eavesdroppers so as to enhance the wireless security of TMA systems. Section V includes numerical results and analyses. Finally, we conclude our work in Section VI.


\emph{Notations}: Throughout the paper, we use boldface uppercase letters, boldface lowercase letters and lowercase letters to denote matrices, column vectors and scalars, respectively. $(\cdot)^T$, $(\cdot)^*$, $(\cdot)^{\dag}$, $(\cdot)^{-1}$, $|\cdot|$, and $\|\cdot\|$ correspond to the transpose, complex conjugate, complex conjugate transpose, inverse, modulus, and $l_2$ norm, respectively. The notation $E\{\}$ denotes the expectation operation and $\boldsymbol{I}$ is the identity matrix.

\section{System model}\label{sysmodel}
We consider a TMA, comprising a phased array with $N$ antennas spaced apart by half wavelength $\lambda_0/2$ (see Fig. \ref{f1}). The transmit waveform is OFDM waveform with  $K$ subcarriers spaced by $f_s$. 
\begin{equation}
{x}(t) = \frac{1}{\sqrt{K}} \sum \limits_{l=1}^{K} s_l e^{j2 \pi [f_0+(l-1)f_s]t}, \quad t\in[0,T_s)
\end{equation}
where $s_l$ is the digitally modulated data symbol assigned to the $l$-th subcarrier,  {$T_s$ is the symbol duration,} $f_0$ denotes the frequency of the first subcarrier and $1/\sqrt{K}$ is the power normalization coefficient that normalizes $s_l$ to be unit power. Note that we eliminate the index of the transmitted OFDM symbol here as the following analyses are independent of the symbol  transmitted. 

{We will assume that the eavesdropper know the channel perfectly and can compensate for its effect. Therefore, the channel will not be shown in the expressions.
If the eavesdropper is unaware of the channel, the system remains secure, preventing the eavesdropper from unscrambling the signal. However, there are situations where the eavesdropper may have knowledge of the channel, such as when the eavesdropper is a friendly node but does not have  authorization to receive the information. Our work specifically addresses this latter scenario, where the eavesdropper is assumed to know the channel.}

\begin{figure}[t]
\centerline{\includegraphics[width=2.4in]{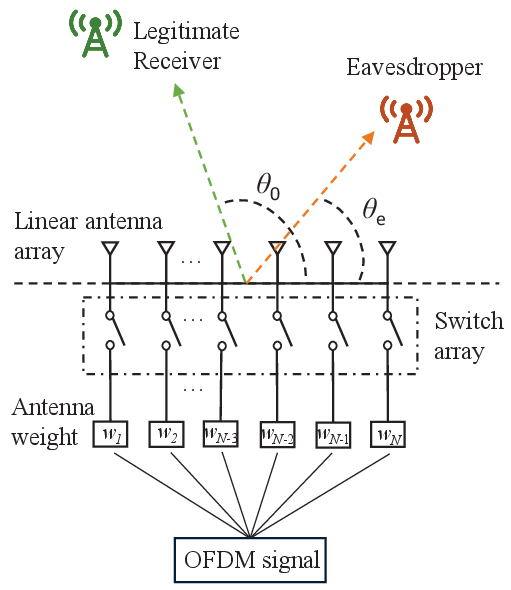}}
\caption{Illustration of the TMA OFDM-enabled DM transmitter.}\label{f1}
\end{figure}

Before being radiated into the half space, $\theta \in [0, \pi]$, the OFDM symbol needs to be multiplied by antenna weights $\{w_n\}_{n=1,2,...,N}$ and manipulated by a ON-OFF temporal function $U(t)$ that controls the switch array periodically. Let the wavelength $\lambda_0$ associate with $f_0$.
The signal radiated by the TMA OFDM system can be expressed as
\begin{equation}
{y}(t, \theta) = \frac{1}{\sqrt{N}} \sum \limits_{n=1}^{N} {x}(t)  w_n  U_n(t)  e^{j(n-1) \pi \cos \theta}.
\end{equation}
In order to focus the beam towards the direction of the legitimate user, $\theta _0$, we set $w_n = e^{-j(n-1) \pi \cos \theta _{\rm 0}}$. The ON-OFF switching function $U_n(t)$ is designed as a periodic square waveform with the time period being $T_s$. On denoting the normalized switch ON  time instant and the normalized ON time duration as $\tau_n^o$ and $\Delta \tau_n$, respectively, we can express $U_n(t)$ as Fourier series as follows: 
\begin{equation}
    U_n(t) = \sum _{m=-\infty}^{\infty} a_{mn} e^{j2 m \pi f_s t},
\end{equation}
where
\begin{equation}\label{sincEq}
    a_{mn} = \Delta \tau _n \sinc (m \pi \Delta \tau _n)  e^{-jm \pi (2 \tau _n^o + \Delta \tau _n)}.
\end{equation}
Here $\sinc (\cdot)$ is an unnormalized sinc function. By combining the above equations, we write the transmitted symbol as
\begin{equation}\label{eq1}
    {y}(t, \theta) = \frac{1}{\sqrt{NK}}  \sum_{l=1}^{K} s_k  e^{j2 \pi [f_0+(l-1)f_s]t}  \sum_{m=-\infty}^{\infty} e^{j2m \pi f_s t}V_m,
\end{equation}
where
\begin{equation}\label{eq2}
    V_m =  \sum \limits_{n=1}^{N} a_{mn} e^{j(n-1) \pi (\cos \theta - \cos \theta _{\rm 0})}.
\end{equation}

Then, the signal seen in direction $\theta$ on the  $i$-th subcarrier equals 
\begin{equation}\label{eq3}
    {y}_i (t, \theta) = \frac{1}{\sqrt{NK}} \sum \limits_{l=1}^{K} s_k e^{j2 \pi [f_0+(l-1)f_s]t} V_{i-l}.
\end{equation}
After OFDM demodulation, the received data symbol on the $i$-th subcarrier can be expressed as $y_i (\theta) = 1/\sqrt{NK} \sum_{l=1}^{K} s_k V_{i-l}$. The received signals on all subcarriers without noises, put in vector $\boldsymbol{y}$, can be expressed as 
\begin{equation}\label{eq4}
    \boldsymbol{y} = \boldsymbol{V} \boldsymbol{s}, 
\end{equation}
where  $\boldsymbol{V} \in \mathbb{C}^{K\times K}$ is a Toeplitz matrix defined  as 
\begin{equation}\label{eq5}
    \boldsymbol{V} = \frac{1}{\sqrt{NK}} \left[ \begin{array}{ccccc}
V_0& V_{-1}& \cdots& V_{-(K-2)}& V_{-(K-1)}\\
V_1& V_0& \cdots& V_{-(K-3)}& V_{-(K-2)}\\
\vdots& \vdots& \ddots& \vdots& \vdots \\
V_{K-2}& V_{K-3}& \cdots& V_0& V_{-1}\\
V_{K-1}& V_{K-2}& \cdots& V_1& V_0
\end{array} \right],
\end{equation}
and $\boldsymbol{s} = [s_1, s_2, \cdots, s_K]^T$. In order to implement  DM functionality, $\tau _n^o$ and $\Delta \tau _n$ must be  chosen  to satisfy $V_{m \ne 0}(\tau _n^o,\Delta \tau _n,\theta = \theta _0) = 0$ and $V_{m = 0}(\tau _n^o,\Delta \tau _n, \theta = \theta _0) \ne 0$. These can be achieved by  the following three conditions \cite{tvt2019time}:
\begin{itemize}
    \item (C1) $\Delta \tau _n , \tau _n^o \in \{\frac{h-1}{N}\}_{h=1,2,...,N}$ (note that the subscript $n$ is not necessarily equal to $h$);
    \item (C2) $\tau _p^o \ne \tau _q^o, \Delta \tau _p = \Delta \tau _q = \Delta \tau$ for $p \ne q$;
    \item  (C3) $\sum_{n=1}^N \Delta \tau _n \ne 0$. 
\end{itemize}  By substituting these three conditions into the above equations, we can find that along $\theta _0$, $\boldsymbol{V}$ is a diagonal matrix and the received OFDM signal  equals
\begin{equation}\label{legireceiver}
    {y}(t, \theta _0) = \Delta \tau \sqrt{\frac{N}{K}} {s}(t).
\end{equation}
In all other directions, the signal of each subcarrier contains the harmonic signals from all other subcarriers, which gives rise to the so-called scrambling,  hence achieving the PLS. Taking into account the additive noise  $\boldsymbol{z} = [z_1, \cdots, z_k, \cdots, z_K]^T$, where $z_k$ is a i.i.d. Gaussian random variable with zero mean and the same variance $\sigma _z^2$, the received signals can be written as
\begin{equation}
    \boldsymbol{y} = \boldsymbol{V} \boldsymbol{s} + \boldsymbol{z}.
\end{equation}

In the following, we will consider the centered and  whitened received signal, i.e.,
\begin{equation}\label{whitening}
\tilde {\boldsymbol{y}}=\boldsymbol{A}\boldsymbol{y}
=\tilde {\boldsymbol{V}}\boldsymbol{s} + \boldsymbol{A} \boldsymbol{z}
\end{equation}
where $\boldsymbol{A}$ is the whitening matrix, obtained based on the eigenvalue decomposition of the covariance  of $\mathbf{y}$ \cite{oja2000ica}, or the quasi-whitening matrix when using Gaussian moments to handle noises \cite{hyvarinen1999gaussian}. For independent, zero-mean, unit-variance inputs $s_k$, and noiseless case, we have
\begin{equation}
 E\{\tilde{\boldsymbol{y}} \tilde{\boldsymbol{y}}^{\dag} \} = \tilde{\boldsymbol{V}} E\{\boldsymbol{s}\boldsymbol{s}^{\dag}\} \tilde{\boldsymbol{V}}^{\dag} = \tilde{\boldsymbol{V}} \tilde{\boldsymbol{V}}^{\dag} = \boldsymbol{I}.
\end{equation}

\section{On Defying  the TMA scrambling by the Eavesdropper}
Let us assume the presence of an eavesdropper in direction $\theta _e$ ($\theta_e \ne \theta _0$). Due to (C1)-(C3),  one can see that, along direction  $\theta _e$, the received OFDM signal on each subcarrier is scrambled by the data symbols modulated onto all other subcarriers, since for $\theta = \theta_e $, $\boldsymbol{V}$ is not diagonal.

Note that $\boldsymbol{y}$ in \eqref{whitening}, contains linear mixtures of the elements of $\boldsymbol{s}$. Both $\boldsymbol{s}$, $\tilde{\boldsymbol{V}}$ are  unknown to the eavesdropper, so the recovery of $\boldsymbol{s}$ can be viewed as a  blind source separation problem. In communications, the elements of $\boldsymbol{s}$ are typically statistically independent with each other and non-Gaussian. Thus, the eavesdropper can leverage an ICA method to estimate $\tilde{\boldsymbol{V}}$, and then, recover $\boldsymbol{V}$ based on $\boldsymbol{A}$.
 
In this section, we first introduce the application of ICA, based on which we propose CMICA, an algorithm for estimating the mixing matrix using short-length data, and then we show how to resolve the ambiguities and fully recover the source signals.

\subsection{The Proposed CMICA for Estimating the Mixing Matrix}
The ICA attempts to recover the mixed   data based on the fact that a linear mixture of independent, non-Gaussian random variables is more Gaussian than the original  variables. Hence, the goal of ICA is to 
find an unmixing matrix $\boldsymbol{W}=[\boldsymbol{w}_1, \cdots, \boldsymbol{w}_l, \cdots, \boldsymbol{w}_K]$ that maximizes the non-Gaussianity of $\boldsymbol{W}^T \tilde{\boldsymbol{y}}$. When $ \boldsymbol{w}_l^T \tilde{\boldsymbol{y}}$ is least Gaussian it is equal to some element of $\boldsymbol{s}$ \cite{oja2000ica}.
 To find more elements of $\boldsymbol{s}$ we need to constrain the
search to the space that gives estimates uncorrelated with the previous ones.

The non-Gaussianity can be quantified via the   kurtosis or the negentropy, {both of which can  be formulated based on each $\boldsymbol{w}_l$ as \cite{bingham2000fast}
\begin{equation} \label{nonGaussianity}
    J_G(\boldsymbol{w}_l) = E\{G(|\boldsymbol{w}_l^{{\dag}}
\tilde{\boldsymbol{y}}|^2)\},
\end{equation}
}
where $G(\cdot)$ is a smooth contrast function, chosen as $  G(v) = \frac{1}{2} v^2$ to approximate  kurtosis, and $G(v) = -\exp(-v/2)$ to approximate negentropy.
Since $\tilde{\boldsymbol{y}}$ is white and zero-mean,  $v$ has  zero mean and unit variance.

%
Here, we need to maximize the sum of $K$ non-Gaussianity quantifiers, one for each subcarrier. We obtain the following constrained optimization problem:
\begin{equation}\label{OPT}
\begin{aligned}
    \underset{{\boldsymbol{w}_1,...,\boldsymbol{w}_K}}{\max} \quad & \sum _{l=1}^{K} J_G(\boldsymbol{w}_l) \\
    \textrm{s.t.} \quad & \boldsymbol{w}_i^{\dag} \boldsymbol{w}_j = \zeta_{ij}, \quad i, j = 1,2,\cdots, K
\end{aligned}
\end{equation}
where $\zeta_{ij} = 1$ for $i=j$ and $\zeta_{ij}=0$ otherwise.
To solve the problem of \eqref{OPT}, we can adopt a FastICA algorithm \cite{bingham2000fast,tao2024tma}, via which the unmixing weights are updated using a fixed-point iteration scheme, {where, in each iteration,  the new weight vector is obtained by the fixed-point algorithm, normalized to unit magnitude, and then ensuring that it is decorrelated from the previously estimated ones via a matrix-based orthogonalization method.}

Even though FastICA converges fast, it needs a large sample set to achieve low estimation error. This is because the non-Gaussianity metric is computed based on the mean of a function of the collected samples, and more samples lead to better mean estimate and better ICA estimates. Good ICA estimates are essential for the subsequent steps of resolving the ambiguities. However, obtaining a large sample of data may not be possible in dynamic communication environments, or in case where the TMA system parameters vary over time. Next, we will show how one could obtain good estimates with a small number of data samples.

In this section, we propose a two-stage method to obtain a (locally) optimum weights. In the first stage, the Newton iteration and the symmetric decorrelation operation are applied to solve the problem of \eqref{OPT}, along the lines of FastICA \cite{bingham2000fast}.
{The result of the first stage is used to initialize the second stage, which follows the same procedure as the first stage, except that the decorrelation operation is omitted.}

\subsubsection{Stage 1} Let us first consider a noiseless case, i.e., $\tilde{\boldsymbol{y}} = \tilde{\boldsymbol{V}} \boldsymbol{s}$, {and for notation simplicity, let $\boldsymbol{w}$ denote any column of matrix $\mathbf{W}$.
The Lagrangian of \eqref{nonGaussianity} under the constraint $E\{|\boldsymbol{w}^{{\dag}}\tilde{\boldsymbol{y}}|^2\} = \|\boldsymbol{w}\|^2 = 1$ is} 
\begin{equation}\label{KTC}
    L(\boldsymbol{w}, \lambda) = E\{G(|\boldsymbol{w}^{{\dag}}\tilde{\boldsymbol{y}}|^2)\} - \lambda (E\{|\boldsymbol{w}^{{\dag}}\tilde{\boldsymbol{y}}|^2\} - 1),
\end{equation}
where $\lambda$ is the Lagrangian multiplier. Adopting an approximate Newton iteration method 
\cite{bingham2000fast}, we obtain an  estimate of  {a (locally) optimum value  of each column of matrix $\boldsymbol{W}$, i.e., $\boldsymbol{w}_{opt}$}, via the following iteration:
\begin{equation} \label{newton}
    \boldsymbol{w} := \boldsymbol{w} - \frac{E\{\tilde{\boldsymbol{y}}(\boldsymbol{w}^{\dag} \tilde{\boldsymbol{y}})^* g (|\boldsymbol{w}^{\dag} \tilde{\boldsymbol{y}}|^2) \} - \lambda \boldsymbol{w}}{E\{g (|\boldsymbol{w}^{\dag} \tilde{\boldsymbol{y}}|^2) + 2|\boldsymbol{w}^{\dag} \tilde{\boldsymbol{y}}|^2 g'(|\boldsymbol{w}^{\dag} \tilde{\boldsymbol{y}}|^2) \} - \lambda},
\end{equation}
where the notation ``:'' denotes the iterative update of $\boldsymbol{w}$;  $\lambda = E\{|\boldsymbol{w}^{\dag} \tilde{\boldsymbol{y}}|^2 g(|\boldsymbol{w}^{\dag} \tilde{\boldsymbol{y}}|^2) \}$; $g$ and $g'$ are the first-order and the second-order derivative of $G$, respectively. 
Subsequently, the mixing matrix $\boldsymbol{W}$, constrcuted based on all estimated  $\boldsymbol{\boldsymbol{w}}$'s, is decorrelated {as follows}:
\begin{equation}\label{decorrelation}
    \boldsymbol{W} := \boldsymbol{W} (\boldsymbol{W}^{\dag} \boldsymbol{W})^{-1/2}.
\end{equation}
The iteration stops when the error between successive estimates of $\boldsymbol{W}$ is below a  threshold, or, when a maximum number of iterations has been completed. {When the sample size is not large enough,} after the Newton iteration stops, the obtained $\boldsymbol{w}$ is only a coarse estimate of $\boldsymbol{w}_{opt}$ since the sample covariance of $\boldsymbol{s}$ will not be equal to the identity matrix, and hence $\tilde{\boldsymbol{V}} \tilde{\boldsymbol{V}}^{\dag} \ne \boldsymbol{I}$. After that point, trying to use decorrelation further cannot yield a better solution.
  %
Beyond that point,  we introduce a fine-tuning stage (stage 2) to refine this coarse estimate.

\subsubsection{Stage 2} We use gradient descent to update the weight vector in this stage, as it provides better control over step size adjustment compared to the Newton method. Using the gradient descent method to maximize  \eqref{KTC} we get the update
\begin{equation} \label{gradient}
    \boldsymbol{w} := \boldsymbol{w} + \mu (E\{\tilde{\boldsymbol{y}}(\boldsymbol{w}^{\dag} \tilde{\boldsymbol{y}})^* g (|\boldsymbol{w}^{\dag} \tilde{\boldsymbol{y}}|^2) \} - \lambda \boldsymbol{w}),
\end{equation}
where $\mu$ is the step size. In this stage we skip the decorrelation operation, and only normalize $\boldsymbol{w}$ after each iteration to satisfy the constraint $\|\boldsymbol{w}\|^2 = 1$. 
{After fine-tuning, we obtain a (locally) optimum demixing matrix, i.e., $\boldsymbol{W}_{opt}$.} 

{The second stage can be viewed as the FastICA  applied to each $\boldsymbol{w}_l$ separately, and as such its convergence is guaranteed \cite{Oja2001ICAbook, hyvarinen1998independent}.}
{
Per \cite{Oja2001ICAbook}, the condition for convergence for the 
 non-Gaussianity objective is
\begin{equation}\label{conforcon}
    E\{g(|s_l|^2) + |s_l|^2 g'(|s_l|^2) - |s_l|^2 g(|s_l|^2)\} < 0,
\end{equation}
where $s_l$ is the $l$-th element of $\boldsymbol{s}$. This 
{is satisfied by choosing appropriate contrast functions like the Kuortosis contrast $G(v) = \frac{1}{2} v^2$.}
 When \eqref{conforcon} is satisfied, the non-Gaussianity quantifier evaluated at $\boldsymbol{w}_l = c \tilde {\boldsymbol{v}}_l$, where $c$ is a complex scaler due to the ICA ambiguity, and $\tilde {\boldsymbol{v}}_l$ is the $l$-th column of $\tilde {\boldsymbol{V}}$, will be larger than  at adjacent points. Based on \eqref{conforcon}, by choosing a small step size, the weight error will keep decreasing and the Stage 2 will converge to a (locally) optimum.}

\begin{algorithm}[t]
    \caption{Proposed CMICA Algorithm}\label{cmica}
    \begin{algorithmic}[1]
        \State Preprocess the collected data $\boldsymbol{y}$ and initialize $\boldsymbol{W}$ randomly;
        \State \textbf{Start the first stage}:
        \ForEach {$i = 1,2,...,K$}
            \State Update $\boldsymbol{w}_i$ according to \eqref{newton};    
        \EndFor
        \State Decorrelate $\boldsymbol{W}$ according to \eqref{decorrelation};
        \State Repeat step 3 $\sim$ 6 until convergence or maximal iteration number;        
        \State \textbf{Start the second stage}:
        \ForEach {$i = 1,2,...,K$}
            \State Update $\boldsymbol{w}_i$ according to \eqref{gradient};
            \State Normalize $\boldsymbol{w}_i$ by $\boldsymbol{w}_i = \boldsymbol{w}_i/\|\boldsymbol{w}_i\|$;
        \EndFor
        \State Repeat step 9 $\sim$ 12 until convergence or maximal iteration number.
    \end{algorithmic}
\end{algorithm}

 The above two-stage ICA method can be applied to any data that satisfy the ICA model. When the source symbols are constant modulus, for example $M$-PSK, which are very common in communication systems, we insert $\boldsymbol{w}_{opt} = \boldsymbol{w}_l = c \tilde {\boldsymbol{v}}_l$ into \eqref{nonGaussianity} and obtain $J_G(\boldsymbol{w}_{opt}) = E\{G(|cs_l|^2)\}$. Computing the expectation using the collected samples, we can easily find that the sample estimate of non-Gaussianity at $\boldsymbol{w}_{opt} = c \tilde {\boldsymbol{v}}_l$ does not change with respect to the length of the sample set under the constant modulus property. Based on this invariance, we could use the above two-stage ICA to find $\boldsymbol{w}_{opt}$ with less data for constant modulus signals, which is shown in the simulations of Section \ref{numerical}.

{In the noisy case, i.e., $\tilde{\boldsymbol{y}} = \tilde{\boldsymbol{V}} \boldsymbol{s} + \boldsymbol{Az}$, the contribution of the colored Gaussian  
noise, $\boldsymbol{A z}$, is suppressed when using kurtosis as the non-Gaussianity metric. When using negentropy as the non-Gaussianity metric, the Gaussian moments-based method \cite{hyvarinen1999gaussian} can be applied to estimate the mixing matrix from noisy data.} We summarize the above two-stage CMICA algorithm in Algorithm 1. {The derivation details of \eqref{newton} and \eqref{gradient} are shown in  Appendix A.}

\subsection{Proposed KNN-Based Scheme for Resolving Ambiguities}


After obtaining $\boldsymbol{W}$,   $(\boldsymbol{WA})^{-1}$ is most probably  not equal to the actual mixing matrix $\boldsymbol{V}$ since there exist scaling and permutation ambiguities in $\boldsymbol{W}$ \cite{oja2000ica}, which would prevent the correct recovery of source symbols.
To resolve these ambiguities, we need to exploit prior knowledge about the TMA OFDM system. Assume that the eavesdropper knows (i) the OFDM specifics of the transmitted signals, like the number of subcarriers, $K$, (ii) the data modulation scheme, (iii) the Toeplitz structure of $\boldsymbol{V}$ and (iv) the rules (C1)-(C3) for implementing TMA. The rules (C1)-(C3) define a set of values for  TMA parameters, therefore, knowledge of the rules does not imply any knowledge of the specific parameters used by the TMA. 

\subsubsection{Resolving the amplitude scaling ambiguity} First, the scaling ambiguity arises because  $\boldsymbol{y} = \boldsymbol{V} \boldsymbol{s}$ can be written  as $\boldsymbol{y} = \sum_{i} (\alpha_i \boldsymbol{V}(:,i)) (s_i / \alpha_i)$ for an arbitrary  $\alpha _i$. The ICA algorithm cannot distinguish between $s_i$ and $s_i / \alpha_i$ since both of them have the same level of non-Gaussianity. Let us separate the scaling ambiguity  into amplitude and phase ambiguity. By knowing the transmit constellation, the eavesdropper  knows the amplitudes of the source signals. As a result, the eavesdropper knows how much the amplitude of the recovered signals is scaled, and can thus recover the amplitude scaling ambiguity. Before resolving the phase ambiguity, the eavesdropper will need to reorder the estimated mixing matrix correctly, which is discussed in the next subsection.

\subsubsection{Resolving the permutation ambiguity} The permutation ambiguity arises because $\boldsymbol{y}$ will not change if the elements of $\boldsymbol{s}$ are permuted and  the  columns of $\boldsymbol{V}$ are accordingly  permuted. Therefore, ICA cannot identify the recovered data symbols in the right order, i.e., it cannot match each demixed data symbol with the right  subcarrier. To solve this issue, we proceed as follows. We define $\boldsymbol{F}\buildrel \triangle \over = (\boldsymbol{WA})^{-1}$. In the absence of ambiguities, $\boldsymbol{F}$ would be equal to $\boldsymbol{V}$ forming a Toeplitz matrix. However, due to the presence of ambiguities, this is not the case.
 We propose reordering 
$\boldsymbol{F}$ by assessing how closely the reordered 
$\boldsymbol{F}$ approximates a Toeplitz matrix structure. Exhaustive reordering is impractical due to the 
$K!$  possible orderings, resulting in prohibitive computational cost.
The reordering process to be explained next 
 has complexity  $O(K^3)$.
{\begin{proposition}\label{prop1}
According to \eqref{eq5}, there are $K$ identical elements in the main diagonal of $\boldsymbol{V}$, and $K-1$, $K-2$, ..., $1$ identical elements in  other diagonals above or under the main diagonal. Since the values of $V_{-(K-1)},\cdots,V_0,\cdots,V_{K-1}$ are different according to \eqref{eq2}, the main diagonal  can solely determine the Toeplitz structure of  $\boldsymbol{V}$. 
Therefore, when 
$\boldsymbol{F}$ is reordered correctly, its main diagonal elements will be nearly identical\footnote{We should note that due to estimation errors within ICA, the estimated diagonal elements will not be exactly the same.}.
\end{proposition}}

\begin{figure}[t]
\centerline{\includegraphics[width=3.2in]{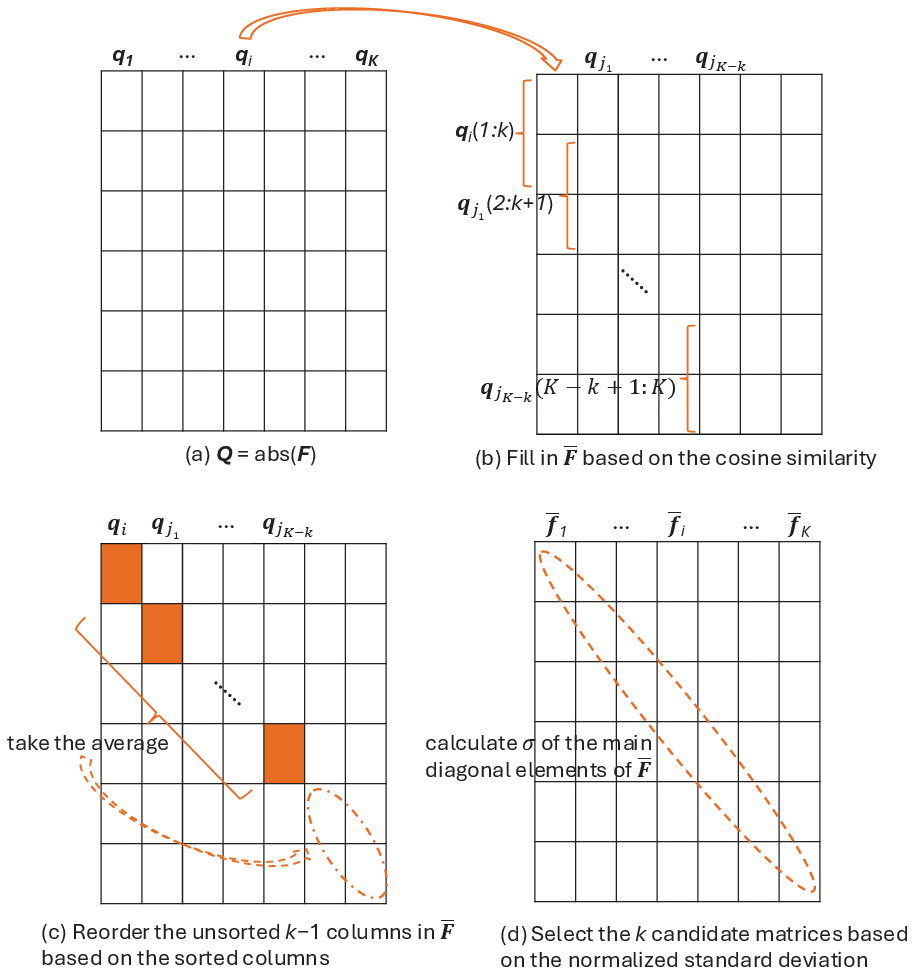}}
\caption{Illustration of the proposed reordering process.}\label{ill_reorder}
\end{figure}

Based on the Proposition \ref{prop1}, we propose to focus only on the main diagonal elements to reorder $\boldsymbol{F}$ and achieve low computational complexity. Specifically, we first calculate the amplitude of each element in $\boldsymbol{F}$ and get a new matrix $\boldsymbol{Q}$, the $i$th column of which is denoted by $\boldsymbol{q}_i$. Then we select the first $k$ elements of $\boldsymbol{q}_i$, i.e., $\boldsymbol{q}_i(1:k)$, as the reference vector, and put $\boldsymbol{q}_i$ in the first column of an empty matrix $\bar{\boldsymbol{F}}$, which is used to store the reordered $\boldsymbol{F}$. Next, we compare $\boldsymbol{q}_{j \ne i}(2:k+1)$ with the selected reference vector based on the cosine similarity and {put the most similar vector}  in the second column of $\bar{\boldsymbol{F}}$. {Cosine similarity is quantified by
$
    \frac{\boldsymbol{q}_i^T(1:k) \boldsymbol{q}_{j \ne i}(2:k+1)}{\|\boldsymbol{q}_i(1:k)\| \|\boldsymbol{q}_{j \ne i}(2:k+1)\|}.
$
}
In turn, we obtain $K-k+1$ reordered columns in $\bar{\boldsymbol{F}}$.
For the remaining unsorted $k-1$ columns in $\boldsymbol{Q}$, we take the average of the main diagonal elements of the matrix formed by those $K-k+1$ reordered columns as the reference, and put these unsorted $k-1$ columns in the  corresponding placements according to the fact that the main diagonal elements of the mixing matrix should be the same.
%
%
For each of the  $K$  reference vectors,  $\boldsymbol{q}_i(1:k), i=1,\cdots,K$, we obtain a matrix $\bar{\boldsymbol{F}}$.
Out of them, we select the $k$  matrices with the least normalized standard deviation of their main diagonal elements, and let the phase ambiguity resolving approach, to be described next, find the most plausible ${\boldsymbol{F}}$, {which are the principles of KNN to resolve the ICA ambiguities. The reordering process is illustrated  in Fig. \ref{ill_reorder}.
\begin{remark}
    Different from a general KNN, the proximity in our proposed KNN-based ambiguity resolving algorithm stems from an artificial Toeplitz similarity, which is measured by the normalized standard deviation of the main diagonal elements and the cosine similarity between the reference and candidate vectors. The use of KNN improves the robustness of the algorithm by retaining $k$-candidate solutions for further evaluation, which mitigates the impact of errors in ICA estimation and noise.
\end{remark}
}


%


\subsubsection{Resolving phase scaling ambiguity} 
At this step, $\boldsymbol{F}$ and  $\boldsymbol{s}$ denote the results of the above described processes that  resolved the  
  amplitude scaling amplitude and permutation ambiguities.
   The phase ambiguity is introduced when $\alpha _i$ is strictly complex. Let us consider  $M$-PSK modulated source symbols\footnote{The extension to QAM modulation is straightforward as QAM is a combination of several kinds of PSK modulation.}. Each source signal can have up to $M$ phase transformations,  hence there are $M$ possible phases for each column of $\boldsymbol{F}$, and  in total $M^K$ phase possibilities for $\boldsymbol{F}$. 
The Toeplitz constraint can reduce the $M^K$ possibilities to $M$; this is because  the phases of the diagonal elements of $\boldsymbol{F}$ must be the same. We define these $M$ possibilities for $\boldsymbol{F}$ as $\boldsymbol{F}_1, \boldsymbol{F}_2, ..., \boldsymbol{F}_M$. {The core principle of resolving the phase ambiguity is shown in the Proposition \ref{prop2}.
\begin{proposition}\label{prop2}
   The phase uncertainty can be eliminated only if there exist a set of TMA parameters, i.e, $N, \Delta \tau, \{\tau _n^o\}_{n=1,2,...,N}, \theta _e, \theta _0$, that correspond to a unique matrix in the set $\{\boldsymbol{F}_u, {u=1,2,...,M}\}$. The feasible TMA parameters must satisfy the rules defined in (C1)-(C3) in Section II and $N \in \mathbb{N}^+$, $\theta _e, \theta _0 \in (0, \pi)$.
\end{proposition}
}

Based on the above proposition, we proceed as follows. From \eqref{sincEq} and \eqref{eq2} we have
\begin{equation} \label{zero}
    V_m = 0, m= \pm N, \pm 2N, \pm 3N, \cdots
\end{equation}
This is because the $\sinc$ term will be 0 when $m= \pm N, \pm 2N, \pm 3N, \cdots$ We can utilize \eqref{zero} to find the value of $N$ from $\boldsymbol{F}$. Then, let $\phi = \cos \theta _e - \cos \theta _0$. From \eqref{eq2} we obtain
\begin{equation}\label{eqV0}
\begin{split}
    V_0 &= \Delta \tau \sum \limits_{n=1}^{N} e^{j(n-1) \pi \phi} = \Delta \tau \frac{1-e^{jN\pi \phi}}{1-e^{j\pi \phi}} \\
    &= \Delta \tau \frac{\sin(\frac{N}{2} \pi \phi)}{\sin(\frac{1}{2} \pi \phi)} e^{j\frac{(N-1)}{2} \pi \phi},
\end{split}
\end{equation}
where $\phi \ne \pm 2/N, \pm 4/N, \pm 6/N, \cdots$ On assuming that $\phi$ is known, which can be obtained via direction finding techniques by the eavesdropper, we will know the actual phase of $V_0$ from \eqref{eqV0}, denoted as $\angle V_0$, since it is determined only by $N$ and $\phi$. Next, for each possible $\boldsymbol{F}_u$ 
we check whether the following holds: 
\begin{equation} \label{angleV0}
    \angle V_0 = \angle \boldsymbol{F}_u(1,1),
\end{equation}
where $\boldsymbol{F}_u(1,1)$ is the main diagonal element of $\boldsymbol{F}_u$. Meanwhile, we need to check whether there exists $\Delta \tau$ constrained by the rules (C1)-(C3) that satisfies
\begin{equation} \label{amVo}
\begin{split}
    &|V_0| = |\boldsymbol{F}_u(1,1)|.
\end{split}
\end{equation}
By \eqref{zero}, \eqref{angleV0} and \eqref{amVo}, we can find the solutions of $N, \Delta \tau, \phi$ for only one of $\{\boldsymbol{F}_u , {u=1,2,...,M}\}$ since there is a fixed phase difference, i.e., $\frac{2\pi}{M}$, between $\boldsymbol{F}_u(1,1)$ and $\boldsymbol{F}_{u+1}(1,1)$. Therefore, the phase ambiguity is resolved when $\phi$ is known.

\begin{algorithm}[t]
    \caption{KNN-based Ambiguity Resolving Algorithm}\label{ARA}
    \begin{algorithmic}[1]
        \State Set $S = \{1,2,...,K\}$;
        \ForEach {$i = 1,2,...,K$}
            \State Take $\boldsymbol{q}_i(1:k)$ as the reference vector and put $\boldsymbol{q}_i$ in the first column of $\bar{\boldsymbol{F}}$;
            \State Remove $i$ from $S$;
            \ForEach {$d = 2,3,...,K-k+1$}
                \State Find $\boldsymbol{q}_j(d:d+k-1)$, $j \in S$ that is the closest to $\boldsymbol{q}_i(1:k)$ based on the cosine similarity;
                \State Put $\boldsymbol{q}_{j}$ in the $d$-th column of $\bar{\boldsymbol{F}}$;
                \State Remove $j$ from $S$;
            \EndFor    
            \State Form a matrix $\boldsymbol{B}_i$ by the above $K-k+1$ reordered columns and take the average of main diagonal of $\boldsymbol{B}_i$ as $b_i$;
            \State Find the remaining $k-1$ unsorted columns based on $b_i$;
            \State Obtain a reordered $\bar{\boldsymbol{F}}_i$ and then take the normalized standard deviation of its main diagonal as $\sigma _i$
        \EndFor
        \State Let $\boldsymbol{\sigma} = [\sigma _1, ..., \sigma _K]$ and select the first $k$ smallest elements in $\boldsymbol{\sigma}$, by which we obtain $k$ reordered $\bar{\boldsymbol{F}}$ and accordingly $k$ reordered $\boldsymbol{F}$;
        \ForEach {reordered $\boldsymbol{F}$}
            \State Obtain $\{\boldsymbol{F}_u, u=1,...,M\}$ according to the transmission constellation and the Toeplitz constraint;
            \ForEach {$\boldsymbol{F}_u$}
                \State Estimate $N$ according to \eqref{zero};
                \If{$\phi$ is known}
                    \State Check if there are solutions to \eqref{angleV0} and \eqref{amVo};
                \Else
                    \State Check if there are solutions to \eqref{eqLambda} and \eqref{angleV0}, \eqref{amVo}, \eqref{Vtau};
                \EndIf
            \EndFor
        \EndFor
    \end{algorithmic}
\end{algorithm}

When $\phi$ is not known, we can proceed as follows. The ratio of the real part and imaginary part of $V_0$, denoted as $\gamma$, is 
\begin{equation}\label{eqLambda}
    \gamma = \frac{1}{\tan{\frac{N-1}{2} \pi \phi}}.
\end{equation}
After estimating $N$ from \eqref{zero}, we can estimate all possible values of $\phi$, $\phi \in (-2,2)$ from \eqref{eqLambda} for each $\{\boldsymbol{F}_u, u=1,...,M\}$. Then, for each found $\phi$, we check if there are solutions to the equation \eqref{angleV0} and \eqref{amVo}. After that, we can  find the feasible values of $N, \Delta \tau$ and $\phi$ for at least one of $\{\boldsymbol{F}_u, u=1,...,M\}$; if two or more $\boldsymbol{F}$ are found, we further check whether there exists $\{\tau _n^o\}_{n=1,2,...,N}$ that satisfies
\begin{equation} \label{Vtau}
    \boldsymbol{V}(\tau _n^o, \Delta \tau, N, \phi) = \boldsymbol{F},
\end{equation}
where $\{\tau _n^o\}_{n=1,2,...,N}$ are subject to (C1)-(C3), and $N, \Delta \tau$ and $\phi$ are the found values by solving \eqref{zero}, \eqref{angleV0} and \eqref{amVo}. Since there are constraints on the feasible TMA parameters as stated above, it is possible to find only one $\boldsymbol{F}$ among $\{\boldsymbol{F}_u, u=1,...,M\}$ by solving the equation \eqref{zero}, \eqref{angleV0}, \eqref{amVo}, \eqref{eqLambda} and \eqref{Vtau} simultaneously. This means  that the eavesdropper could still eliminate the phase uncertain even when $\phi$ is not known. {Considering that $\{\tau _n^o\}_{n=1,2,...,N}$ are discrete and there are no analytical solutions to \eqref{Vtau}, we adopt an exhaustive search method here to find $\{\tau _n^o\}_{n=1,2,...,N}$ with the complexity $O(N!)$, which is much higher than that of when $\phi$ are known. So the eavesdropper can use direction finding techniques to estimate $\phi$ first and then resolve the phase ambiguity when $N$ is large.
\begin{remark}
    Another difference from the traditional KNN is that our proposed KNN-based ambiguity resolving algorithm uses equations \eqref{zero}-\eqref{Vtau} to determine the most plausible solution by testing candidates against physical and structural constraints specific to the TMA system. This process leverages prior knowledge about the system, such as the rules (C1)-(C3) and the Toeplitz constraint of $\boldsymbol{V}$, rather than relying on the majority class of the $k$-nearest candidates.
\end{remark}
}
We summarize the whole ambiguity resolving algorithm in Algorithm 2.

\section{Defending the TMA scrambling}


\subsection{Conditions for A Secure TMA}

There are two scenarios where the TMA OFDM system is sufficiently secure: when 
$\boldsymbol{V}$ is rank-deficient and when there is non-uniqueness in the ON-OFF switching pattern. In the former case, there are multiple solutions for 
 $\boldsymbol{s}$ when solving the problem of $\boldsymbol{y} = \boldsymbol{V} \boldsymbol{s}$.
 The non-uniqueness means that there are multiple ON-OFF switching patterns, i.e., multiple groups of 
 $\Delta \tau $ and $\{\tau _n^o\}_{n=1,2,...,N}$ when trying to defy the scrambling.

\begin{figure}[t]
\centering
\subfigure[]{
\centering
\includegraphics[width=2.9in]{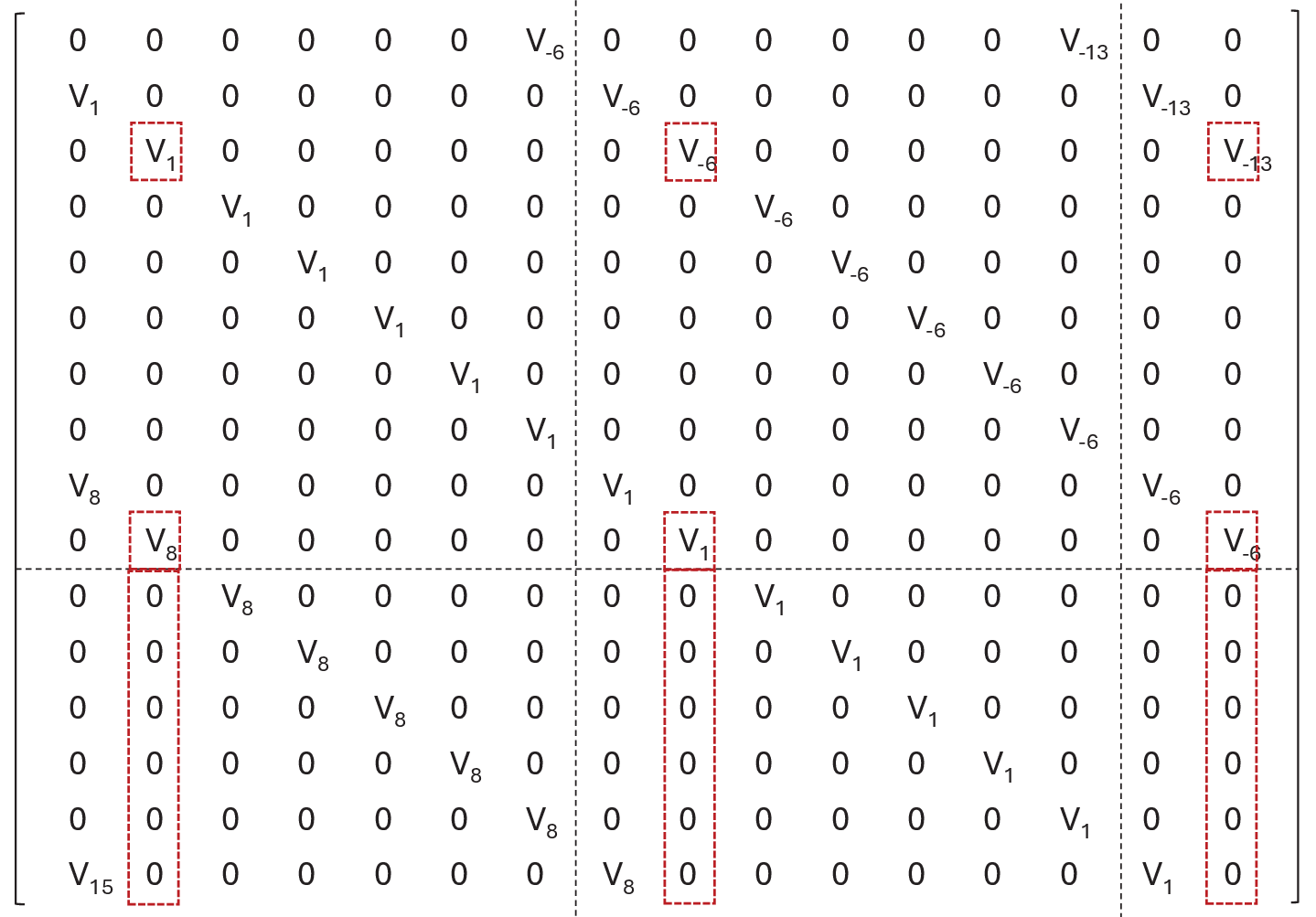}
}

\subfigure[]{
\centering
\includegraphics[width=3.1in]{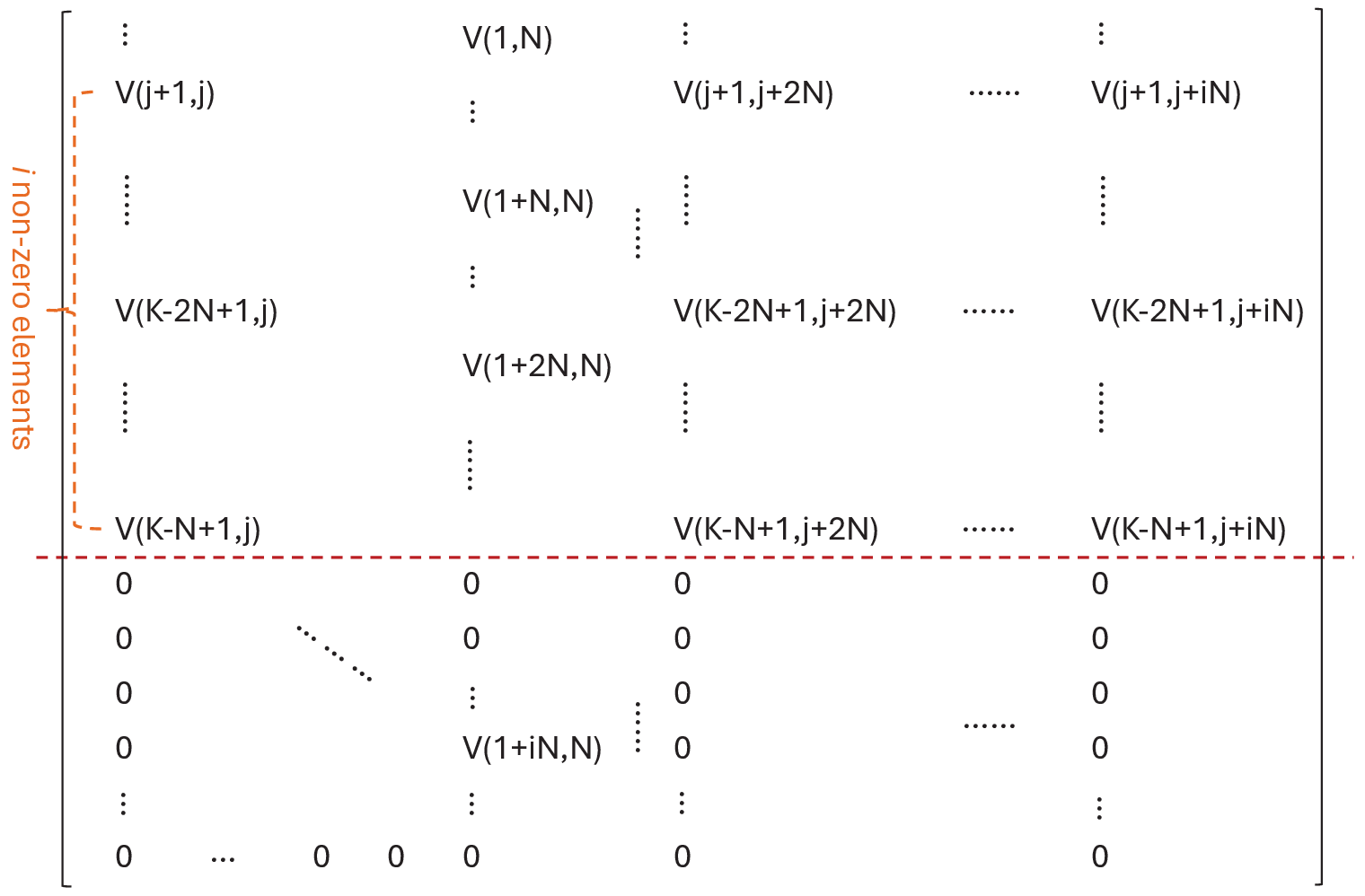}
}
\centering
\caption{Illustration of $\boldsymbol{V}$ when $\phi= \pm 2/N, \pm 4/N, \pm 6/N, \cdots$ and $\tau _n^o = (n-1)/N$: (a) an specific example with $K=16$, $N=7$; (b) a general diagram used for proving Lemma \ref{le1}.}
\label{proofNR}
\end{figure}

When $\phi = \cos \theta _e - \cos \theta _{\rm 0} = \pm 2/N, \pm 4/N, \pm 6/N, \cdots$, the aforementioned two  cases become feasible. Specifically, on using  $\phi = \pm 2/N, \pm 4/N, \pm 6/N, \cdots$ and $\tau _n^o = (n-1)/N$ in \eqref{eq2}, we get
\begin{equation} \label{rank}
V_m 
    \begin{dcases}
     \ne 0  & m = 1 + iN, \, i=0, \pm 1, \pm 2, \cdots \\
      = 0  & \text{otherwise}
    \end{dcases}
\end{equation}
Based on the above observations, we obtain the following corollary.
\begin{corollary}\label{coro1}
The distance between the indices corresponding to the nearest two non-zero elements in the same row or column of $\boldsymbol{V}$ is $N$.
\end{corollary}
\noindent We show an example of this special mixing matrix for $N = 7, K=16$ in Fig. \ref{proofNR}(a). For this kind of mixing matrix, we can prove the following lemma:
\begin{lemma}\label{le1}
    $\boldsymbol{V}$, as defined in \eqref{eq5}, is not full-rank when $iN \ne K, \forall i \in \mathbb{N^+}$.
\end{lemma}
\begin{proof}

Suppose $iN \ne K, \forall i \in \mathbb{N^+}$. Let us set $K = iN + j$, where $j$ is the reminder, $1 \le j \le N-1$.
Let us divide the matrix $\boldsymbol{V}$ into two parts: $\boldsymbol{V}^+ \in \mathbb{C}^{K-N+1 \times K}$ and $\boldsymbol{V}^- \in \mathbb{C}^{N-1 \times K}$, as shown in Fig. \ref{proofNR} (b); the parts  are separated by the red dotted line in Fig. \ref{proofNR} (b).
According to Corollary \ref{coro1}, the $N$-th column of $\boldsymbol{V}$ has $i+1$ non-zero elements. Also,  its $(i+1)$-th non-zero element is $\boldsymbol{V}(1+iN, N)$, and is located in $\boldsymbol{V}^-$ since $K-N+2 \le 1+iN \le K$. 

 Along the diagonal including $\boldsymbol{V}(1+iN, N)$, we can find one non-zero element $\boldsymbol{V}(K-N+1, j)$. This is because, the elements onto the diagonal including $\boldsymbol{V}(1+iN, N)$ are all the same and hence non-zero according to the Toeplitz constraint. For the non-zero element located in the $K-N+1$-th row, its index of column is $N - ((1+iN)-(K-N+1)) = K - iN = j$. So we get the non-zero element $\boldsymbol{V}(K-N+1, j)$.

 For the $j$-th column of $\boldsymbol{V}$, we  know that its $N-1$ elements located in $\boldsymbol{V}^-$ are all 0, as shown in Fig. \ref{proofNR} (b) and according to Corollary \ref{coro1}. For its part located in $\boldsymbol{V}^+$, there are $i$ non-zero elements since $K-N+1 = (i-1)N + j + 1$ and it has the non-zero elements $\boldsymbol{V}(j+1, j)$, $\boldsymbol{V}(N+j+1, j)$, $\cdots$, $\boldsymbol{V}((i-1)N + j + 1, j)$. We have similar results for the $(j+N)$-th, $(j+2N)$-th, $\cdots$, $(j+iN)$-th (exactly the $K$-th) column. Therefore, we have $i+1$ columns in $\boldsymbol{V}$ that have only $i$ non-zero elements located in the $(j+1)$-th, $(N+j+1)$-th, $\cdots$, $((i-1)N + j + 1)$-th rows of $\boldsymbol{V}^+$ and the elements located in $\boldsymbol{V}^-$ are all 0.

 According to the Leibniz formula for finding the determinant of a matrix \cite{anton2013elementary}, after selecting $i$ non-zero elements located in different rows out of the corresponding $i+1$ columns, there must be one column left. For the column left, its row indices of non-zero elements have been occupied, so there must be 0 existing in the Leibniz formula. Therefore, the determinant of $\boldsymbol{V}$ is zero and  Lemma \ref{le1} is proved.
\end{proof}

\noindent As shown in Fig. \ref{proofNR} (a), $K=16$ and $N=7$, so there is no $i \in \mathbb{N^+}$ that satisfies $K=iN$. From this figure, we can see there exist three columns that have only two non-zero elements located in $\boldsymbol{V}^+$, which are marked by the red boxes.

When $\phi= \pm 2/N, \pm 4/N, \pm 6/N, \cdots$ and $\tau _n^o \ne (n-1)/N$, there exist possibly multiple groups of $\Delta \tau $, $\{\tau _n^o\}_{n=1,2,...,N}$ and $\boldsymbol{s}$ that correspond to the same $\boldsymbol{y}$. In fact, it is intractable to solve for multiple group of $\Delta \tau $, $\{\tau _n^o\}_{n=1,2,...,N}$ and $\boldsymbol{s}$ directly from $\boldsymbol{y} = \boldsymbol{V} \boldsymbol{s}$ after incorporating $\phi= \pm 2/N, \pm 4/N, \pm 6/N, \cdots$ into it, since $\boldsymbol{V}$ contains many $\sin$ and $\exp$ terms and it is a transcendental equation. To shed light on the reason, taking $\phi= 2/N, N=K=4$ and BPSK modulation symbol [-1, +1], we can obtain two different groups of $\Delta \tau $, $\{\tau _n^o\}_{n=1,2,...,N}$ and $\boldsymbol{s}$, i.e., $1/N$, $\{3/N, 1/N, 2/N, 0 \}$, $[+1, -1, -1, +1]$ and $3/N$, $\{0, 2/N, 3/N, 1/N \}$, $[-1, +1, +1, -1]$ that both correspond to the same $\boldsymbol{y}$. Upon further examination of this example, we find that the non-uniqueness of the ON-OFF switching pattern is due to the periodicity and parity properties of the $\sin$ and $\exp$ terms (where the $\exp$ term can be represented in terms of $\sin$ and $\cos$) in $\boldsymbol{V}$.

\subsection{Measures for Defending the Scrambling} \label{defending}
We can enhance the wireless security of the TMA transmitter by rotating it at a certain angle $\theta _r$ to satisfy $\cos (\theta _e + \theta _r) - \cos (\theta _0 + \theta _r) = \pm 2/N, \pm 4/N, \pm 6/N, \cdots$. In this case we need to know the eavesdropper location. Based on the aforementioned conditions for a secure TMA, it is impossible for an eavesdropper to apply our proposed defying scheme to resolve the ambiguity when $\phi= \pm 2/N, \pm 4/N, \pm 6/N, \cdots$. Furthermore, the eavesdropper cannot crack the TMA OFDM system completely by any means under the first class of condition, as the system is underdetermined, thereby ensuring sufficient security. 

{We should note that by rotating the TMA transmitter  the signal-to-noise ratio (SNR) at the legitimate receiver is not affected, since the TMA system is still subject to the system configurations shown in Section \ref{sysmodel} after rotation. The received signal still satisfies \eqref{legireceiver} and hence the SNR is ${\frac{N \Delta \tau^2 }{K \sigma_z^2}}$, which is independent with the direction of the legitimate receiver.}
 
We can also design mechanisms to defend the TMA scrambling against the eavesdroppers by exploiting the  need of ICA to operate in a stationary  stationary environment.
We can disturb the applicability of ICA by changing the mixing matrix of TMA over time.  This can be done by selecting randomly $\{\tau _n^o\}_{n=1,2,...,N}$ in each OFDM symbol period according to $\tau _n^o \in \{\frac{h-1}{N}\}_{h=1,2,...,N}$ and $\tau _p^o \ne \tau _q^o$. Meanwhile this mechanism is able to maintain the DM functionality as it still satisfies the scrambling scheme. The cost is that this will increase the hardware design complexity since it requires the switch ON-OFF pattern changing frequently. Moreover, we can degrade ICA by disturbing the independence of source signals, which can be achieved by randomly assigning some identical symbols to be transmitted on multiple subcarriers but it will result in lower bit rate. These two methods do not require knowledge of the eavesdropper location.

\linespread{1.3}
\begin{table}[t]
	\begin{center}
	\caption{Main Notations and Definitions}\label{tab2}
	\begin{tabular}{p{1.6cm}<{\centering}|p{5.4cm}<{\centering}}
	\hline
        \hline
        \textbf{Notation} & \textbf{Definition}\\
        \hline
	$N$ & The number of antenna elements \\
	\hline
	$K$ & The number of subcarriers \\
        \hline
	$H$ & The number of used OFDM symbols in ICA \\
        \hline
        $\boldsymbol{V}$ and $V_m$ & The mixing matrix defined in \eqref{eq5} and its element defined in \eqref{eq2}\\
	\hline
	$\Delta \tau$ and $\tau_n^o$ & The normalized ON time duration and the normalized switch ON  time instant \\
        \hline
        (C1)-(C3) & The rules for choosing $\Delta \tau$ and $\tau_n^o$ \\
        \hline
        $\theta _0$ and $\theta _e$ & The direction of the legitimate user and the direction of the eavesdropper \\
        \hline
        $\phi$ & The difference between $\cos \theta _e$ and $\cos \theta _0$ \\
        \hline
        $\hat J _G$ & The sample estimate of the non-Gaussianity metric \\
        \hline
        $k$ & The length of reference vector used in the KNN-based ambiguity resolving algorithm \\
        \hline
        	
	\end{tabular}
	\end{center}
\end{table}
\linespread{1.16}

\section{NUMERICAL RESULTS}\label{numerical}
In this section, we present numerical results to evaluate our proposed TMA scrambling defying and defending schemes. First of all, we summarize the main parameters and their definitions used in the simulations in Table \ref{tab2} in order to enhance the readability of the paper. Then, we simulate a TMA OFDM-enabled DM system with $N = 7$ antennas as the same as in \cite{tvt2019time}. We set the TMA parameters according to the rules (C1)-(C3) and adopt the BPSK modulation. Also, We use BER as the performance metric to evaluate the proposed approaches. For the receiver noise, we adopt kurtosis as the non-Gaussianity metric, and kurtosis, as a higher-order statistics, can inherently mitigate the noise effect \cite{hyvarinen1999gaussian}. The value of $\phi$ is assumed to be known by the eavesdropper and $k=3$ unless otherwise specified. For other parameters, they are specified in the corresponding experiments. The results and analyses are as follows.

\linespread{1.3}
\begin{table}[t]
    \begin{center}
    \caption{Average BER of the TMA OFDM-enabled DM system}\label{tab1}
    \begin{tabular}{p{0.3cm}<{\centering}|p{0.5cm}<{\centering}|p{0.5cm}<{\centering}|p{1.58cm}<{\centering}|p{1.4cm}<{\centering}|p{1.73cm}<{\centering}}
    \hline
    \hline
    No. & $\theta _0$($^\circ$) & $\theta _e$($^\circ$) & Original BER & Defied BER & Defended BER   \\ \hline
    1   & 50    & 90          & 0.4964      & 0  & 0.4618   \\ 
    2   & 60    & 30          & 0.4952      & 0  & 0.4020   \\ 
    3   & 80    & 40          & 0.4948      & 0  & 0.4951   \\ 
    4   & 90    & 50          & 0.4891      & 0  & 0.4834  \\ 
    5   & 100   & 80          & 0.4824      & 0  & 0.4142   \\ \hline
    6   & 30    & 70          & 0.5210      & 0  & 0.5103   \\ 
    7   & 40    & 90          & 0.4742      & 0  & 0.4675   \\ 
    8   & 50    & 130         & 0.4934      & 0  & 0.4883   \\ 
    9   & 80    & 150         & 0.4876      & 0  & 0.4750   \\ 
    10  & 90    & 140         & 0.4607      & 0  & 0.4631   \\ \hline
    \end{tabular}
    \end{center}
\end{table}
\linespread{1.16}

\subsection{Effectiveness of the Proposed Scrambling Defying and Defending Schemes}
We simulated a TMA OFDM scenario with   $K = 16$ OFDM subcarriers, $H = 1e4$ data samples, and conducted $10$ groups of experiments, where $\theta _0$ and $\theta _e$ were chosen differently in each experiment as shown in Table \ref{tab1}. 
In experiments 1$-$5, $\theta _0$ and $\theta _e$ were taken as known by the eavesdropper when resolving the phase ambiguity, while in experiments 6$-$10, $\theta _0$ and $\theta _e$ were taken as unknown  and $\phi$ was estimated according to \eqref{eqLambda}. For each experiment, we set the  SNR as 20 dB and $\Delta \tau = (N-1)/N$, $\{\tau _n^o\}_{n=1,2,...,N} = (n-1)/N$.
The BER results are shown in Table \ref{tab1}.
In the table, `Original BER' denotes the BER at $\theta _e$ based on the raw signals received by the eavesdropper, while `Defied BER' denotes  the BER based on the recovered signals via the proposed defying scheme, and `Defended BER' denotes the BER after  applying the 
 defending mechanism to confront the proposed defying scheme.
The mechanism applied here is that  of changing $\{\tau _n^o\}_{n=1,2,...,N}$ randomly in each OFDM symbol period.  From 
Table \ref{tab1} we can see that the eavesdropper experiences non-zero original BER  due to the  TMA scrambling. {In all cases, the defied BER is $0$ \footnote{{Here, ``BER = 0'' in our experiments reflects the ability of the proposed defying algorithm to fully resolve the TMA scrambling under given conditions. It does not imply that BER = 0 would occur in a real-world implementation, where achieving an exact BER of 0 is impossible due to practical imperfections, as our work focuses on demonstrating the feasibility of defying TMA scrambling and does not account for all practical impairments.}}, meaning that the eavesdropper is able to defy the scrambling completely and correctly recover the transmitted source signals.} Also, in all cases, the defended BER is not $0$, demonstrating that the proposed mechanism takes effect in defending the TMA scrambling and enhancing the system security.

\begin{figure}[t]
\centerline{\includegraphics[width=2.9in]{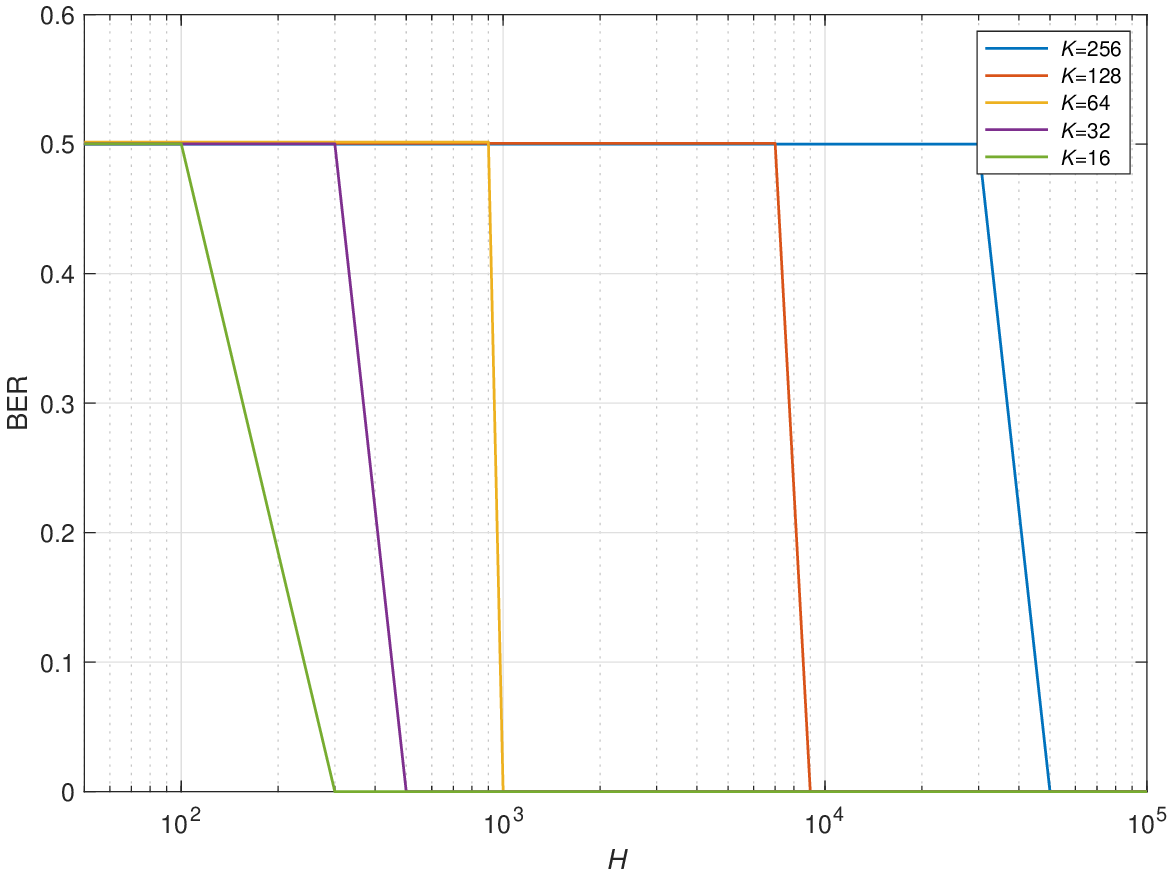}}
\caption{The defying performance of the proposed scheme for different values of number of subcarriers ($K$)  and data length ($H$).}\label{fig3}
\end{figure}

The defying performance of our proposed defying scheme for different values of $K$ and $H$ is shown in Fig. \ref{fig3}. In this figure, we set $\theta _0 = 60^{\circ}$, $\theta _e = 40^{\circ}$, $\Delta \tau = 1/N$, $\{\tau _n^o\}_{n=1,2,...,N} = (n-1)/N$ and $\phi$ is taken as known. The SNR is set as $20$ dB. From Fig. \ref{fig3}, we can observe that the defying performance improves with $H$, as expected, and the BER can be reduced to $0$ even when $K = 256$, demonstrating  the great scalability of our proposed scheme. Moreover, it can be seen that the defying scheme requires many more samples  when $K$ is large. This is because a larger $K$ corresponds to  a larger number of source signals, and thus ICA needs more samples to work well. Additionally, when $K$ is large and $|m|$ close to $K$, $V_m$ is  very small due to the term $\sinc (m \pi \Delta \tau _n)$ in (\ref{eq2}). Considering that there are also   estimation errors in ICA, for large $K$ and $|m|$ close to $K$, $|V_m|$   could be even smaller than the estimation errors of ICA, which will eventually lead to  failure of the ambiguity resolving algorithm. Therefore, a large number of samples are needed to improve the accuracy of ICA estimates and accordingly the performance of the ambiguity resolving procedure when facing a large $K$.

\begin{figure}[t]
\centerline{\includegraphics[width=2.95in]{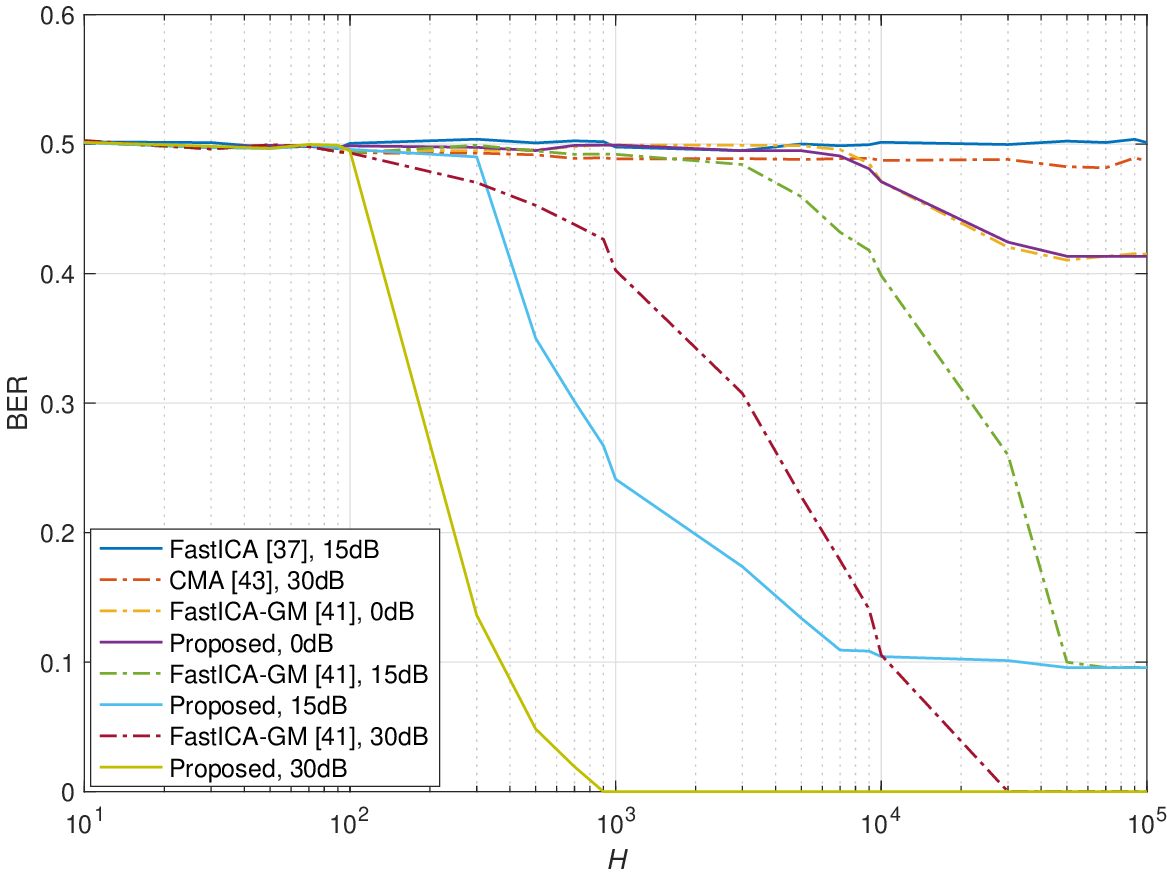}}
\caption{Comparison of the scrambling defying efficiency among CMICA and benchmarks for different values of $H$ and SNR.
}\label{fig4}
\end{figure}

\subsection{Efficiency of the Proposed Defying Scheme}
{Here, we compare the scrambling defying performance of the proposed CMICA algorithm against several benchmark methods under different SNR values and numbers of OFDM symbols. The benchmarks include the classical FastICA without noise removal \cite{bingham2000fast}, FastICA with Gaussian moments-based noise mitigation (FastICA-GM) \cite{hyvarinen1999gaussian}, and the constant modulus algorithm (CMA) \cite{johnson1998blind}. 
CMA is particularly chosen for comparison because, like CMICA, it is designed to estimate the mixing matrix blindly for constant modulus signals. For the system configurations, the parameters $\theta _0 = 60^{\circ}$, $\theta _e = 40^{\circ}$, $K = 16$ were used, and 30 different sets of TMA parameters, $\Delta \tau $ and $\{\tau _n^o\}_{n=1,2,...,N}$, are generated randomly according to the rules (C1)-(C3). The SNR values were set at 30 dB, 15 dB, and 0 dB, respectively. The resulting BERs, averaged on these different groups of TMA parameters, are shown in Fig. \ref{fig4}. From the figure, it can be observed that the BERs of both the FastICA-based and CMA-based defying schemes remain almost constant regardless of the used number of samples. This is because FastICA without noise mitigation cannot effectively address the noise, and CMA inherently struggles to distinguish different source symbols despite utilizing the constant modulus property. In contrast, the CMICA-based defying scheme and the FastICA-GM-based defying scheme exhibit a gradual reduction in BER as $H$ increases, particularly at high SNR. Notably, the BER of CMICA decreases much faster and more significantly than that of FastICA-GM, highlighting its superior sample utilization efficiency. For instance, with 30 dB SNR, CMICA requires approximately $10^3$ samples to defy scrambling completely, i.e., BER = 0, whereas FastICA-GM requires more than $10^4$. 
At lower SNR = $15$ dB, a similar trend can be observed, albeit with reduced performance, while at 
$0$ dB, 
there is no performance gap between CMICA and FastICA-GM, and their BER reductions occur at almost the same rate. This demonstrates that the performance advantage of CMICA diminishes and its effectiveness converges with FastICA-GM in noisy environments, as expected.}

\begin{figure}[t]
\centerline{\includegraphics[width=2.9in]{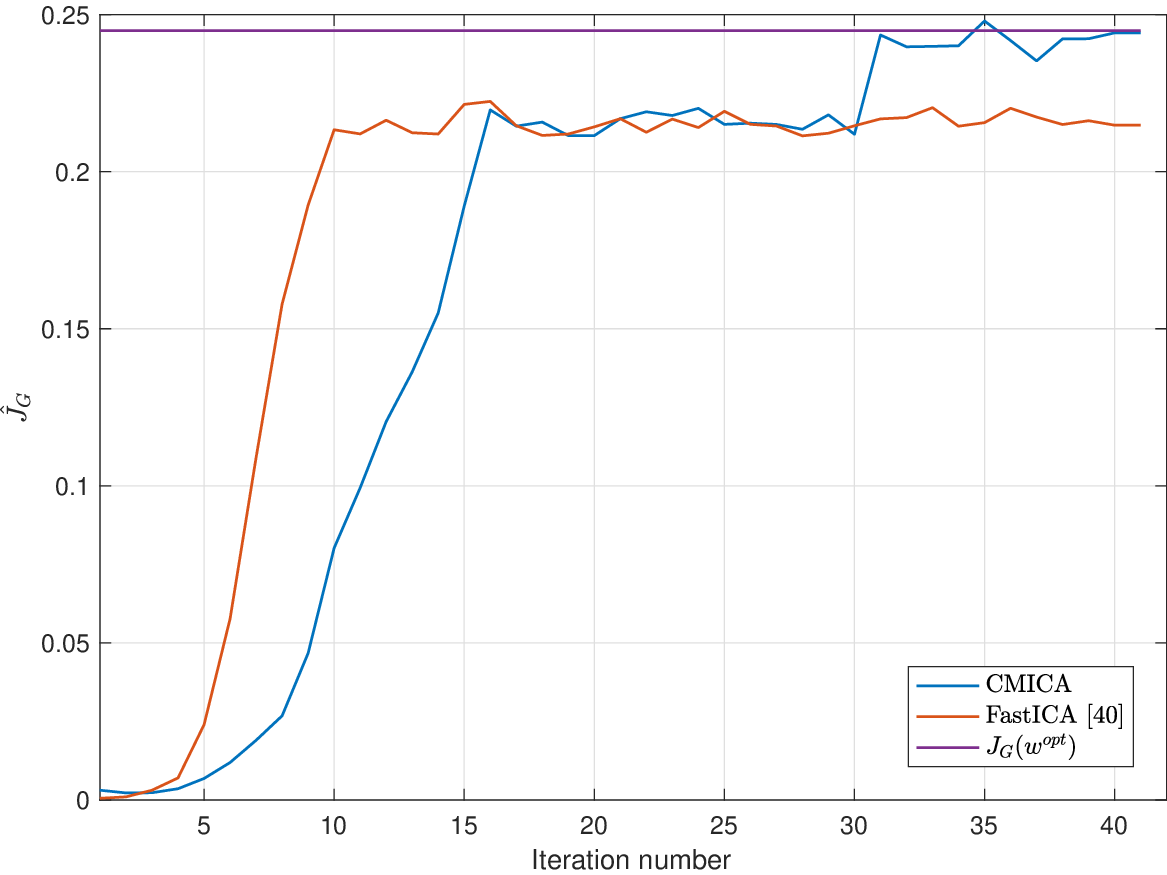}}
\caption{Sample estimate of non-Gaussianity with respect to the ICA iteration number.}\label{nonGaussianFig}
\end{figure}

In Fig. \ref{nonGaussianFig}, we analyze the behavior of the proposed CMICA method during the two-stage iteration process when using a small number of samples. Specifically, we compare the sample estimates of total non-Gaussianity for CMICA and FastICA \cite{bingham2000fast} over the course of iterations, under the following experimental settings: $K = 16$, $\theta _0 = 60^{\circ}$, $\theta _e = 40^{\circ}$, $\Delta \tau = 1/N$, $\{\tau _n^o\}_{n=1,2,...,N} = (n-1)/N$, with $H=1,000$ samples and a maximum of $40$ iterations.
From Fig. \ref{nonGaussianFig}, we observe that CMICA undergoes two distinct convergence phases. In the first stage, the Newton iteration rapidly converges to a solution, but the resulting non-Gaussianity deviates from $J_G(\boldsymbol{w}_{opt})$. In the second stage, the gradient descent iteration refines the solution, enabling CMICA to achieve a larger value of $\tilde{J}_G(\boldsymbol{w})$ that is near $J_G(\boldsymbol{w}_{opt})$. This indicates that the found $\boldsymbol{w}$ in the end is close to $\boldsymbol{w}_{opt}$. In contrast, FastICA fails to reach $J_G(\boldsymbol{w}_{opt})$ under the same conditions. 
This failure arises from the reliance of FastICA on decorrelation, which assumes sufficient statistical independence of the source signals. With limited data, this assumption is violated, leading to suboptimal convergence. These results further demonstrate that the two-stage approach of CMICA, which omits decorrelation in the second stage, can achieve superior performance in estimating the unmixing matrix for small sample scenarios.

\begin{figure}[t]
\centering
\subfigure[]{
\centering
\includegraphics[width=2.9in]{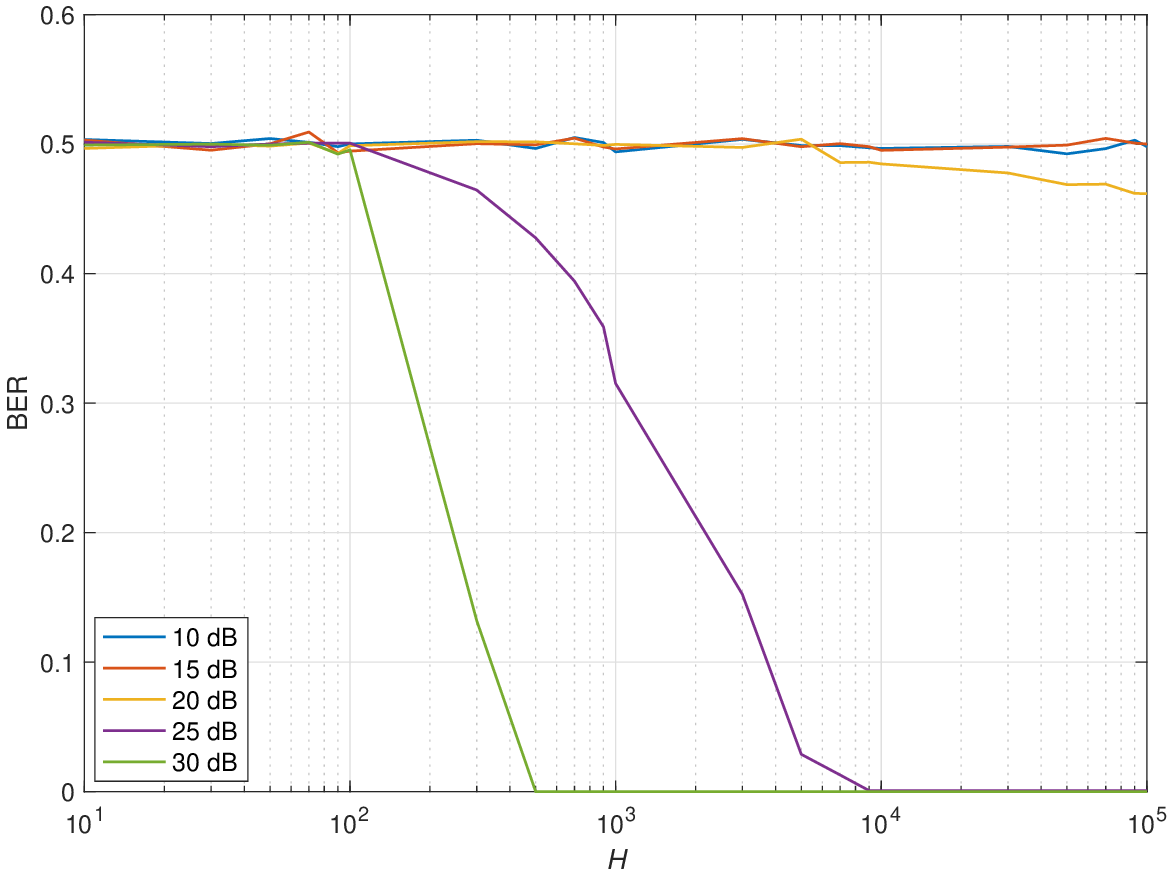}
}

\subfigure[]{
\centering
\includegraphics[width=2.9in]{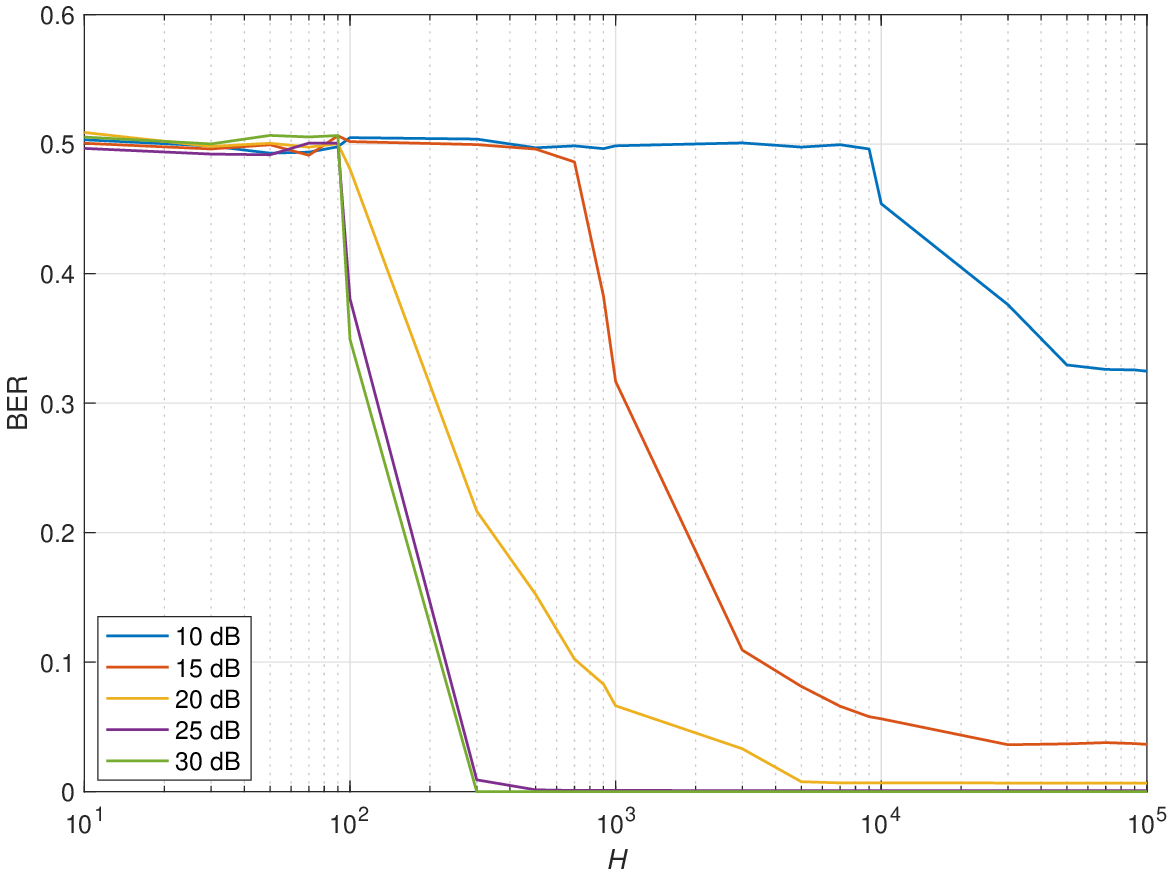}
}
\centering
\caption{The scrambling defying performance with respect to the number of samples and the SNR: a) $k=1$; b) $k=3$ in the KNN-based ambiguity resolving scheme.}
\label{ARAcomparison}
\end{figure}

\subsection{Robustness of the Proposed Defying Scheme}
{Next, we demonstrate the robustness of the proposed KNN-based ambiguity resolving algorithm to the ICA estimation errors in Fig. \ref{ARAcomparison}, focusing on the effect of varying the reference vector length $k$. The experimental configurations include $K=16$, $\theta _0 = 60^{\circ}$, $\theta _e = 40^{\circ}$, $\Delta \tau = 1/N$, $\{\tau _n^o\}_{n=1,2,...,N} = (n-1)/N$, with SNR values set to $(10, 15, 20, 25, 30)$ dB and varying numbers of OFDM symbols $H$. Fig. \ref{ARAcomparison} (a) shows the results for $k = 1$, while Fig. \ref{ARAcomparison} (b) shows the results for $k = 3$. From both figures, it can be observed that the defied BER declines more rapidly at higher SNR values as $H$ increases, while the performance degrades at lower SNRs, as expected. Even with a large $H$, the scrambling defying performance deteriorates significantly at lower SNRs, indicating that noise effect cannot be completely mitigated by kurtosis and larger ICA estimation errors induced by the lower SNR lead to worse defying performance. Also, notice that the effect of $k$ is evident when comparing Fig. \ref{ARAcomparison} (a) and Fig. \ref{ARAcomparison} (b). Especially at lower SNRs (e.g., 10–20 dB), the defying performance for $k = 3$ is significantly better than for $k = 1$. The improved robustness for larger $k$ arises from the principle of KNN, where increasing the number of reordered mixing matrices as candidate solutions provides greater resilience to noise. For $k = 1$, a single candidate solution is used in phase ambiguity resolution, which is highly susceptible to ICA estimation errors. In contrast, when $k = 3$, multiple candidate solutions are considered, allowing the ambiguity resolving algorithm to exclude plausible but incorrect solutions. Therefore, increasing $k$ can improve the defying robustness by leveraging the KNN principle to mitigate the impact of ICA estimation errors on the reordering and phase ambiguity resolving processes.}

In Fig. \ref{snrFig}, we demonstrate the robustness of the proposed scrambling defying method to receiver noise and compare its performance with the benchmark \cite{bingham2000fast}, the original TMA OFDM-enabled DM system \cite{tvt2019time}, and the optimal defying scheme. {The optimal defying scheme means that the eavesdropper is assumed to have perfect knowledge of the actual mixing matrix, allowing it to bypass the TMA scrambling completely; in this case, the BER is affected only by receiver noise. The original TMA system, on the other hand, refers to the original system without any scrambling defying scheme applied. In this experiment, we set $K=16$, $\theta _0 = 60^{\circ}$, $\theta _e = 40^{\circ}$, $\Delta \tau = (N-1)/N$, $\{\tau _n^o\}_{n=1,2,...,N} = (n-1)/N$, $H = 1e4$, and various SNR levels ranging from -50 dB to 50 dB. From Fig. \ref{snrFig}, we can observe again that the proposed defying method with $k = 3$ outperforms both the $k = 1$ configuration and the benchmark \cite{bingham2000fast}, exhibiting superior robustness to noise, as previously explained. Also, the defied BER of the proposed scheme decreases sharply with increasing SNR, contrasting with the gradual decline of the optimal defied BER, and the performance gap between the proposed method and the optimal scheme at low SNR levels is large, indicating the sensitivity of the proposed defying scheme to high noise levels. One reason is that larger noise levels lead to higher ICA estimation errors, and these errors directly impact the mixing matrix reordering performance \footnote{In fact, the ICA estimation error affects the resolving performance of amplitude scaling ambiguity first and then the reordering performance. We omit to put the amplitude scaling ambiguity here since resolving the permutation and phase scaling ambiguity is much more challenged and their effects on the overall defying performance are more significant.}, which in turn affects the phase ambiguity resolving performance. The resulting \textit{error propagation} amplifies receiver noise and makes the proposed scheme more susceptible to noise.}

\begin{figure}[t]
\centerline{\includegraphics[width=2.9in]{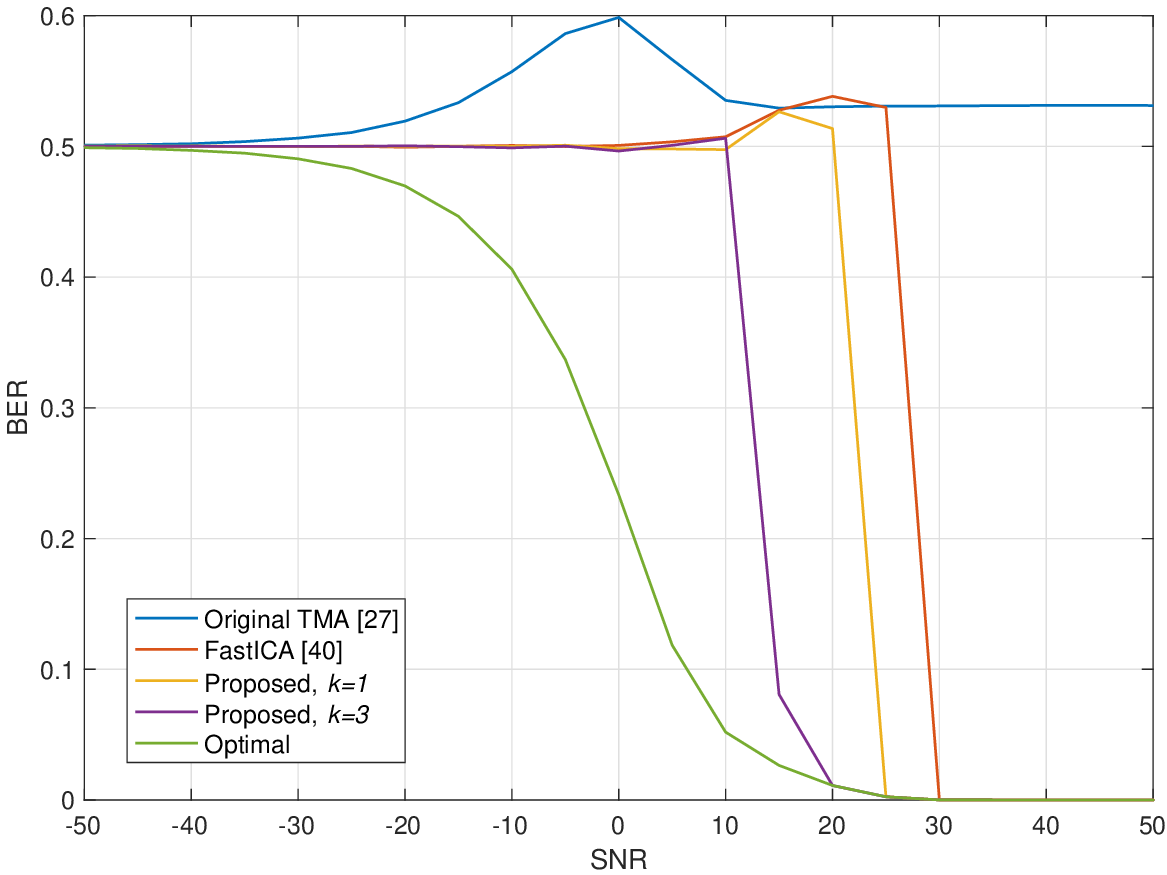}}
\caption{Comparison of the defied BER vs. different SNR among our proposed method, the benchmark and the optimal defying.}\label{snrFig}
\end{figure}

{Moreover, an interesting phenomenon can be observed in Fig. \ref{snrFig}: the BER increases with SNR in specific ranges for all methods except the optimal scheme before eventually declining. For example, the BER increases from -30 to 0 dB for the original TMA system, from 10 to 15 dB for the proposed method with $k = 1$, etc. This behavior is caused by the superposition of noise and TMA scrambling effects. In the case of the optimal defying scheme, which is affected solely by noise, the BER decreases gradually as SNR increases, following the expected trend in general communication systems. However, for the original TMA system, the impact of TMA scrambling becomes significant as SNR increases. At very low SNRs, noise dominates, and the BER remains around 0.5. As the SNR increases, the TMA scrambling effect starts to take over, leading to a temporary increase in BER. When the SNR becomes sufficiently high and noise is negligible, the BER converges to a value that reflects the sole impact of TMA scrambling. For the proposed defying scheme, which are capable of defying TMA scrambling, the superposition effect is less pronounced. While their BER also increases in low-SNR regions due to residual scrambling effects, their ability to mitigate scrambling ensures that the increase is less severe compared to the original TMA system. This behavior highlights again the their effectiveness in mitigating the effects of TMA scrambling under different noise levels.}

\subsection{Trade-off between Power Efficiency and Security}

\begin{figure}[t]
\centerline{\includegraphics[width=2.9in]{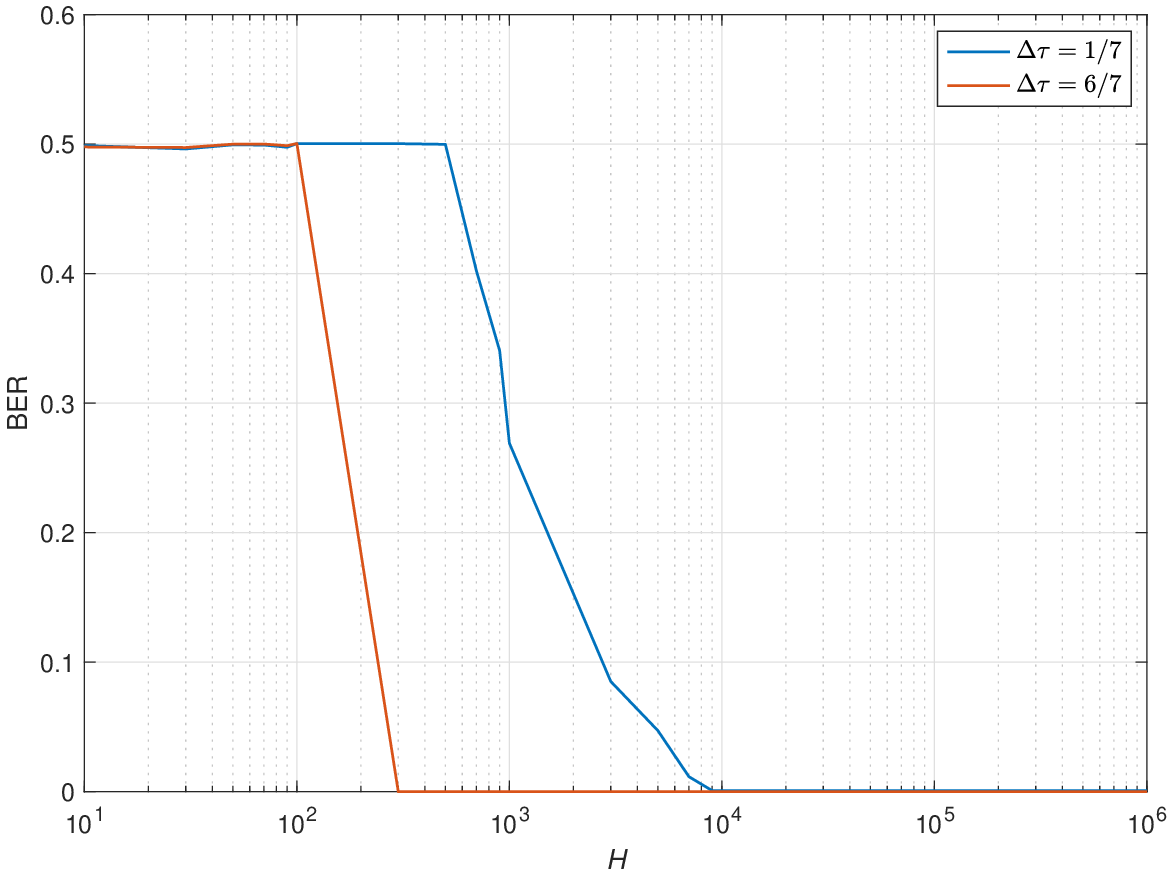}}
\caption{Defied BER of the proposed defying scheme with different $\Delta \tau$.}\label{pes}
\end{figure}

{Finally, we compare the BER performance of the proposed defying scheme under different $\Delta \tau$ values in Fig. \ref{pes}. In this experiment, we set $K=16$, $\theta _0 = 60^{\circ}$, $\theta _e = 40^{\circ}$, $k=3$, SNR = 30 dB, and generate $\{\tau _n^o\}_{n=1,2,...,N} = (n-1)/N$ randomly according to (C1)-(C3). From Fig. \ref{pes}, we can observe clearly that the proposed scheme with $\Delta \tau = 1/7$ exhibits poorer performance compared to that with $\Delta \tau = 6/7$, as the former requires significantly more OFDM samples to defy the scrambling completely, implying that the TMA OFDM-enabled DM system with a smaller $\Delta \tau$ is more secure. This observation aligns with intuition, as $\Delta \tau$ directly affects the power efficiency of the TMA system. As noted in [27], the power efficiency is given by $\Delta \tau^2 \times 100\%$. When $\Delta \tau$ is smaller, the power efficiency decreases substantially, resulting in a significant reduction in the effective SNR. This reduction enhances the system scrambling security against the proposed defying scheme but comes at the cost of wasted power. Consequently, $\Delta \tau$ should be chosen carefully to balance the power efficiency and security in the TMA OFDM-enabled DM systems.}

\section{Conclusion}
DM via TMAs transmitting OFDM waveforms has been viewed as an emerging hardware-efficient and low-complexity approach to secure wireless mobile communication systems. In this paper, we have presented, for the first time, a comprehensive assessment and analysis of wireless security of this kind of system. First, we have shown that this DM transmitter is not secure enough from the perspective of eavesdroppers. Specifically, we have formulated the defying of the TMA scrambling as an ICA problem for the eavesdropper, and shown that under certain conditions the ICA ambiguities can be resolved by exploiting prior knowledge about the TMA OFDM system. For the ICA part, we have proposed an efficient ICA method, namely CMICA, that applied to constant modulus signals and works well for short data lengths.  For the ambiguities, we construct a KNN-based resolving algorithm by exploiting jointly the Toeplitz structure of the mixing matrix, knowledge of data constellation, and the rules for designing the TMA ON-OFF pattern, etc. Then, we have showcased two kinds of conditions, for which the TMA OFDM systems are secure enough, and proposed some mechanisms that can be used to defend the scrambling against the attack of eavesdroppers. Through numerical results and analyses, we have demonstrated the effectiveness, efficiency, and robustness of our proposed defying and defending schemes in the end. Future studies will consider the extension of CMICA to other scenarios with constant-modulus signals. Also, the proposed defying scheme is promising to implement multiple-user DM simultaneously considering that the original TMA OFDM transmitter supports only single-user DM at a time.

\begin{appendices}
\section{Derivations of \eqref{newton} and \eqref{gradient}}
{According to the Lagrange multiplier method, a (local) optimum of \eqref{nonGaussianity} under the constraint $E\{|\boldsymbol{w}^{{\dag}}\tilde{\boldsymbol{y}}|^2\}= \boldsymbol{w}^{{\dag}} E\{\tilde{\boldsymbol{y}} \tilde{\boldsymbol{y}}^{\dag}\} \boldsymbol{w} = \|\boldsymbol{w}\|^2 = 1$ (note that $E\{\tilde{\boldsymbol{y}} \tilde{\boldsymbol{y}}^{\dag}\} = \boldsymbol{I}$ after whitening) are obtained when $\nabla L(\boldsymbol{w}, \lambda) = 0$, i.e.,
\begin{equation}\label{LG}
    \nabla E\{G(|\boldsymbol{w}^{{\dag}}\tilde{\boldsymbol{y}}|^2)\} - \lambda \nabla E\{|\boldsymbol{w}^{{\dag}}\tilde{\boldsymbol{y}}|^2\} = 0,
\end{equation}
where $\lambda = E\{|\boldsymbol{w}_{opt}^{\dag} \tilde{{\boldsymbol{y}}}|^2 g(|\boldsymbol{w}_{opt}^{\dag} \tilde{{\boldsymbol{y}}}|^2) \}$  is the Lagrangian multiplier, and $\boldsymbol{w}_{opt}$ is a (locally) optimum $\boldsymbol{w}$.} The gradient is computed with respect to real and imaginary parts of $\boldsymbol{w}$, respectively. For the left hand of \eqref{LG}, we have
\begin{equation}
\begin{split}
    \nabla E\{G(|\boldsymbol{w}^{{\dag}}\tilde{\boldsymbol{y}}|^2)\} &= 
    2 \begin{bmatrix} 
    E\{{\rm Re}\{y_1(\boldsymbol{w}^{{\dag}}\tilde{\boldsymbol{y}})^*\} g(|\boldsymbol{w}^{{\dag}}\tilde{\boldsymbol{y}}|^2) \} \\ 
    E\{{\rm Im}\{y_1(\boldsymbol{w}^{{\dag}}\tilde{\boldsymbol{y}})^*\} g(|\boldsymbol{w}^{{\dag}}\tilde{\boldsymbol{y}}|^2) \} \\ 
    \vdots \\ 
    E\{{\rm Re}\{y_K(\boldsymbol{w}^{{\dag}}\tilde{\boldsymbol{y}})^*\} g(|\boldsymbol{w}^{{\dag}}\tilde{\boldsymbol{y}}|^2) \} \\
    E\{{\rm Im}\{y_K(\boldsymbol{w}^{{\dag}}\tilde{\boldsymbol{y}})^*\} g(|\boldsymbol{w}^{{\dag}}\tilde{\boldsymbol{y}}|^2) \} 
    \end{bmatrix},\\
    \\
    &= 2 E\{\tilde{\boldsymbol{y}} (\boldsymbol{w}^{{\dag}}\tilde{\boldsymbol{y}})^* g(|\boldsymbol{w}^{{\dag}}\tilde{\boldsymbol{y}}|^2) \},
\end{split}
\end{equation}
and
\begin{equation}
    \nabla E\{|\boldsymbol{w}^{{\dag}}\tilde{\boldsymbol{y}}|^2\} = 
    2 \begin{bmatrix} 
    {\rm Re}\{w_1\} \\ 
    {\rm Im}\{w_1\} \\ 
    \vdots \\ 
    {\rm Re}\{w_n\} \\ 
    {\rm Im}\{w_n\} 
    \end{bmatrix} E\{\tilde{\boldsymbol{y}} \tilde{\boldsymbol{y}}^{\dag}\} = 
    2 \boldsymbol{w}.
\end{equation}
Then we obtain
\begin{equation}\label{fistorder}
\begin{split}
    & \nabla E\{G(|\boldsymbol{w}^{{\dag}}\tilde{\boldsymbol{y}}|^2)\} - \lambda \nabla E\{|\boldsymbol{w}^{{\dag}}\tilde{\boldsymbol{y}}|^2\} = \\
    & 2 E\{\tilde{\boldsymbol{y}} (\boldsymbol{w}^{{\dag}}\tilde{\boldsymbol{y}})^* g(|\boldsymbol{w}^{{\dag}}\tilde{\boldsymbol{y}}|^2) \} - 2 \lambda \boldsymbol{w}.
\end{split}
\end{equation}
Next, 
\begin{equation}
\begin{split}
    & \nabla^2 E\{G(|\boldsymbol{w}^{{\dag}}\tilde{\boldsymbol{y}}|^2)\} = \\
    & 2 E\{\tilde{\boldsymbol{y}} \tilde{\boldsymbol{y}}^{\dag} g(|\boldsymbol{w}^{{\dag}}\tilde{\boldsymbol{y}}|^2) + 2 \tilde{\boldsymbol{y}} \tilde{\boldsymbol{y}}^{\dag} (\boldsymbol{w}^{{\dag}}\tilde{\boldsymbol{y}})^* (\boldsymbol{w}^{{\dag}}\tilde{\boldsymbol{y}}) g'(|\boldsymbol{w}^{{\dag}}\tilde{\boldsymbol{y}}|^2) \} \boldsymbol{I},
\end{split}
\end{equation}
Utilizing the approximation
\begin{equation}
    E\{\tilde{\boldsymbol{y}} \tilde{\boldsymbol{y}}^{\dag} g(|\boldsymbol{w}^{{\dag}}\tilde{\boldsymbol{y}}|^2)\} = E\{\tilde{\boldsymbol{y}} \tilde{\boldsymbol{y}}^{\dag}\} E\{g(|\boldsymbol{w}^{{\dag}}\tilde{\boldsymbol{y}}|^2)\},
\end{equation}
we have
\begin{equation}
\begin{split}
    & \nabla^2 E\{G(|\boldsymbol{w}^{{\dag}}\tilde{\boldsymbol{y}}|^2)\} = \\
    & 2 E\{g(|\boldsymbol{w}^{{\dag}}\tilde{\boldsymbol{y}}|^2)\} + 2 \boldsymbol{w}^{\dag} \tilde{{\boldsymbol{y}}}|^2 g'(|\boldsymbol{w}^{\dag} \tilde{{\boldsymbol{y}}}|^2) \} \boldsymbol{I}.
\end{split}
\end{equation}
Moreover,
\begin{equation}
    \lambda \nabla ^2 E\{|\boldsymbol{w}^{{\dag}}\tilde{\boldsymbol{y}}|^2\} = 2 \lambda \boldsymbol{I}.
\end{equation}
So, the Jacobian of the left-hand part of \eqref{LG} is
\begin{equation}\label{secondorder}
\begin{split}
    & \nabla^2 E\{G(|\boldsymbol{w}^{{\dag}}\tilde{\boldsymbol{y}}|^2)\} - \lambda \nabla ^2 E\{|\boldsymbol{w}^{{\dag}}\tilde{\boldsymbol{y}}|^2\} =\\
    & (2 E\{g(|\boldsymbol{w}^{{\dag}}\tilde{\boldsymbol{y}}|^2)\} + 2 \boldsymbol{w}^{\dag} \tilde{{\boldsymbol{y}}}|^2 g'(|\boldsymbol{w}^{\dag} \tilde{{\boldsymbol{y}}}|^2) \} - 2\lambda)\boldsymbol{I}.
\end{split}
\end{equation}
{Meanwhile, we approximate $\lambda$ using the current value of $\boldsymbol{w}$ instead of $\boldsymbol{w}_{opt}$, i.e.,
\begin{equation}\label{updatelambda}
    \lambda = E\{|\boldsymbol{w}^{\dag} \tilde{{\boldsymbol{y}}}|^2 g(|\boldsymbol{w}^{\dag} \tilde{{\boldsymbol{y}}}|^2) \}.
\end{equation}
}Therefore, based on \eqref{fistorder}, \eqref{secondorder}, \eqref{updatelambda}, we can obtain \eqref{newton} and \eqref{gradient}.

\end{appendices}

\bibliography{Jour24_Ref}

\begin{thebibliography}{10}
\providecommand{\url}[1]{#1}
\csname url@samestyle\endcsname
\providecommand{\newblock}{\relax}
\providecommand{\bibinfo}[2]{#2}
\providecommand{\BIBentrySTDinterwordspacing}{\spaceskip=0pt\relax}
\providecommand{\BIBentryALTinterwordstretchfactor}{4}
\providecommand{\BIBentryALTinterwordspacing}{\spaceskip=\fontdimen2\font plus
\BIBentryALTinterwordstretchfactor\fontdimen3\font minus \fontdimen4\font\relax}
\providecommand{\BIBforeignlanguage}[2]{{%
\expandafter\ifx\csname l@#1\endcsname\relax
\typeout{** WARNING: IEEEtran.bst: No hyphenation pattern has been}%
\typeout{** loaded for the language `#1'. Using the pattern for}%
\typeout{** the default language instead.}%
\else
\language=\csname l@#1\endcsname
\fi
#2}}
\providecommand{\BIBdecl}{\relax}
\BIBdecl

\bibitem{Wyner1975Wire}
A.~D. Wyner, ``The wire-tap channel,'' \emph{Bell Syst. Tech. J}, vol.~54, no.~8, pp. 1355--1387, Aug. 1975.

\bibitem{poor2017wireless}
H.~V. Poor and R.~F. Schaefer, ``Wireless physical layer security,'' \emph{Proceedings of the National Academy of Sciences}, vol. 114, no.~1, pp. 19--26, Jan. 2017.

\bibitem{qiu2023decomposed}
B.~Qiu, W.~Cheng, and W.~Zhang, ``Decomposed and distributed directional modulation for secure wireless communication,'' \emph{IEEE Tran. on Wire. Commun.}, vol.~23, no.~5, pp. 5219--5231, May 2023.

\bibitem{daly2009dire}
M.~P. Daly and J.~T. Bernhard, ``Directional modulation technique for phased arrays,'' \emph{IEEE Tran. on Ante. and Prop.}, vol.~57, no.~9, pp. 2633--2640, Sep. 2009.

\bibitem{su2021secure}
N.~Su, F.~Liu, and C.~Masouros, ``Secure radar-communication systems with malicious targets: {Integrating} radar, communications and jamming functionalities,'' \emph{IEEE Tran. on Wire. Commun.}, vol.~20, no.~1, pp. 83--95, 2022.

\bibitem{Xiao2023Synthesis}
J.~M. Purushothama, Y.~Ding, G.~Goussetis, G.~Huang, and Y.~Xiao, ``Synthesis of energy efficiency-enhanced directional modulation transmitters,'' \emph{IEEE Trans. on Green Comm. and Net.}, vol.~7, no.~2, pp. 635--648, June 2023.

\bibitem{dong2009improving}
L.~Dong, Z.~Han, A.~P. Petropulu, and H.~V. Poor, ``Improving wireless physical layer security via cooperating relays,'' \emph{IEEE Trans. on Signal Processing}, vol.~58, no.~3, pp. 1875--1888, 2010.

\bibitem{Li2011cooperative}
J.~Li, A.~P. Petropulu, and S.~Weber, ``On cooperative relaying schemes for wireless physical layer security,'' \emph{IEEE Trans. on Signal Processing}, vol.~59, no.~10, pp. 4985--4997, 2011.

\bibitem{Li2020relay}
Q.~Li and L.~Yang, ``Beamforming for cooperative secure transmission in cognitive two-way relay networks,'' \emph{IEEE Trans. Inf. Forensics Security}, vol.~15, pp. 130--143, 2020.

\bibitem{zhang2019AN}
W.~Zhang, J.~Chen, Y.~Kuo, and Y.~Zhou, ``Artificial-noise-aided optimal beamforming in layered physical layer security,'' \emph{IEEE Commun. Lett.}, vol.~23, no.~1, pp. 72--75, 2019.

\bibitem{wang2017AN}
W.~Wang, K.~C. Teh, and K.~H. Li, ``Artificial noise aided physical layer security in multi-antenna small-cell networks,'' \emph{IEEE Trans. Inf. Forensics Security}, vol.~12, no.~6, pp. 1470--1482, 2017.

\bibitem{su2022secure}
N.~Su, F.~Liu, Z.~Wei, Y.-F. Liu, and C.~Masouros, ``Secure dual-functional radar-communication transmission: {Exploiting} interference for resilience against target eavesdropping,'' \emph{IEEE Trans. on Wireless Communications}, vol.~21, no.~9, pp. 7238--7252, 2022.

\bibitem{Ottersten2016}
A.~Kalantari, M.~Soltanalian, S.~Maleki, S.~Chatzinotas, and B.~Ottersten, ``Directional modulation via symbol-level precoding: {A} way to enhance security,'' \emph{IEEE J. Sel. Topics Signal Process.}, vol.~10, no.~8, pp. 1478--1493, Dec. 2016.

\bibitem{Alodeh2016DM}
M.~Alodeh, S.~Chatzinotas, and B.~Ottersten, ``Energy-efficient symbol-level precoding in multiuser {MISO} based on relaxed detection region,'' \emph{IEEE Tran. on Wire. Commun.}, vol.~15, no.~5, pp. 3755--3767, 2016.

\bibitem{khandaker2018secure}
M.~R. Khandaker, C.~Masouros, K.-K. Wong, and S.~Timotheou, ``Secure {SWIPT} by exploiting constructive interference and artificial noise,'' \emph{IEEE Trans. on Communications}, vol.~67, no.~2, pp. 1326--1340, 2018.

\bibitem{khandaker2018constructive}
M.~R. Khandaker, C.~Masouros, and K.-K. Wong, ``Constructive interference based secure precoding: {A} new dimension in physical layer security,'' \emph{IEEE Trans. on Information Forensics and Security}, vol.~13, no.~9, pp. 2256--2268, 2018.

\bibitem{valliappan2013antenna}
N.~Valliappan, A.~Lozano, and R.~W. Heath, ``Antenna subset modulation for secure millimeter-wave wireless communication,'' \emph{IEEE Trans. on Communications}, vol.~61, no.~8, pp. 3231--3245, 2013.

\bibitem{ding2017Circular}
Y.~Ding, V.~Fusco, and A.~Chepala, ``Circular directional modulation transmitter array,'' \emph{IET Microw., Antennas and Propag.}, vol.~11, no.~3, pp. 1909--1917, Oct. 2017.

\bibitem{Hamdi2016subset}
N.~N. Alotaibi and K.~A. Hamdi, ``Switched phased-array transmission architecture for secure millimeter-wave wireless communication,'' \emph{IEEE Trans. Commun.}, vol.~64, no.~3, pp. 1303--1312, Mar. 2016.

\bibitem{ding2017free}
Y.~Ding and V.~Fusco, ``A synthesis-free directional modulation transmitter using retrodirective array,'' \emph{IEEE J. Sel. Topics Signal Process.}, vol.~11, no.~2, pp. 428--441, Mar. 2017.

\bibitem{tvt2024security}
G.~Huang, S.~Chen, Y.~Ding, X.~Li, A.~Nallanathan, and S.~Mumtaz, ``Security-enhanced directional modulation symbol synthesis using high efficiency time-modulated arrays,'' \emph{IEEE Trans. Veh. Tech.}, vol.~73, no.~1, pp. 1418--1423, Jan. 2024.

\bibitem{manica2009almost}
L.~Manica, P.~Rocca, L.~Poli, and A.~Massa, ``Almost time-independent performance in time-modulated linear arrays,'' \emph{IEEE Antennas and Wireless Propagation Letters}, vol.~8, pp. 843--846, Jul. 2009.

\bibitem{Massa20144d}
P.~Rocca, Q.~Zhu, E.~T. Bekele, S.~Yang, and A.~Massa, ``4-d arrays as enabling technology for cognitive radio systems,'' \emph{IEEE Trans. Antennas Propag.}, vol.~62, no.~3, pp. 1102--1116, Mar. 2014.

\bibitem{Massa2018time}
J.~Guo, L.~Poli, M.~A. Hannan, P.~Rocca, S.~Yang, and A.~Massa, ``Time-modulated arrays for physical layer secure communications: {Optimization}-based synthesis and experimental assessment,'' \emph{IEEE Trans. Antennas Propag.}, vol.~66, no.~12, pp. 6939--6949, Dec. 2018.

\bibitem{Kummer1963new}
W.~H. Kummer, A.~T. Villeneuve, and F.~G. Terrio, ``New antenna idea - {Scanning} without phase shifters,'' \emph{Electronics}, vol.~36, no.~1, p. 27–32, Mar 1963.

\bibitem{Kummer1963Ultra}
W.~H. Kummer, A.~T. Villeneuve, T.~S. Fong, and F.~G. Terrio, ``Ultra low sidelobes from time-modulated arrays,'' \emph{IEEE Trans. Antennas Propag.}, vol.~11, no.~6, pp. 633--639, 1963.

\bibitem{tvt2019time}
Y.~Ding, V.~Fusco, J.~Zhang, and W.~Wang, ``Time-modulated ofdm directional modulation transmitters,'' \emph{IEEE Trans. Veh. Tech.}, vol.~68, no.~8, pp. 8249--8253, Aug. 2019.

\bibitem{Wu2022Metamaterial}
S.~Vosoughitabar, A.~Nooraiepour, W.~Bajwa, N.~Mandayam, and C.~Wu, ``Metamaterial-enabled 2d directional modulation array transmitter for physical layer security in wireless communication links,'' in \emph{2022 IEEE/MTT-S International Microwave Symposium}, Denver, CO, 2022.

\bibitem{vosoughitabar2023directional}
S.~Vosoughitabar, A.~Nooraiepour, W.~U. Bajwa, N.~B. Mandayam, and C.-T.~M. Wu, ``Directional modulation retrodirective array-enabled physical layer secured transponder for protected wireless data acquisition,'' in \emph{2023 IEEE/MTT-S International Microwave Symposium-IMS}, San Diego, CA, 2023, pp. 1180--1183.

\bibitem{nooraiepour2023programming}
A.~Nooraiepour, S.~Vosoughitabar, C.-T.~M. Wu, W.~U. Bajwa, and N.~B. Mandayam, ``Programming wireless security through learning-aided spatiotemporal digital coding metamaterial antenna,'' \emph{Adv. Intell. Syst.}, vol.~5, no.~10, 2023.

\bibitem{shan2021multi}
G.~Huang, Y.~Ding, and S.~Ouyang, ``Multicarrier directional modulation symbol synthesis using time-modulated phased arrays,'' \emph{IEEE Trans. Antennas Propag.}, vol.~20, no.~4, pp. 567--571, Apr. 2021.

\bibitem{shan2022target}
G.~Huang, Y.~Ding, S.~Ouyang, and J.~M. Purushothama, ``Target localization using time-modulated directional modulated transmitters,'' \emph{IEEE Sensors J.}, vol.~22, no.~13, pp. 13\,508--13\,518, Jul. 2022.

\bibitem{zhaoyi2022TMA}
Z.~Xu and A.~Petropulu, ``A secure dual-function radar communication system via time-modulated arrays,'' in \emph{Proc. IEEE RadarConf'23}, San Antonio, TX, 2023.

\bibitem{xu2024tma}
------, ``Time-modulated intelligent reflecting surface for waveform security,'' in \emph{Proc. IEEE Int. Conf. Acoust., Speech Signal Process. (ICASSP)}, Seoul, Korea, 2024.

\bibitem{nooraiepour2022time}
A.~Nooraiepour, S.~Vosoughitabar, C.-T.~M. Wu, W.~U. Bajwa, and N.~B. Mandayam, ``Time-varying metamaterial-enabled directional modulation schemes for physical layer security in wireless communication links,'' \emph{ACM Journal on Emerging Technologies in Computing Systems}, vol.~18, no.~4, pp. 1--20, Oct. 2022.

\bibitem{li2022chaotic}
H.~Li, Y.~Chen, and S.~Yang, ``Chaotic-enabled phase modulation in time-modulated arrays for secure transmission,'' \emph{IEEE Trans. Antennas Propag.}, vol.~70, no.~11, pp. 10\,454--10\,464, Nov. 2022.

\bibitem{bingham2000fast}
E.~Bingham and A.~Hyvarinen, ``A fast fixed-point algorithm for independent component analysis of complex valued signals,'' \emph{Int. J. Neural Syst.}, vol.~10, pp. 1--8, 2000.

\bibitem{tao2024tma}
Z.~Tao, Z.~Xu, and A.~Petropulu, ``How secure is the time-modulated array-enabled ofdm directional modulation?'' in \emph{Proc. IEEE Int. Conf. Acoust., Speech Signal Process. (ICASSP)}, Seoul, Korea, 2024.

\bibitem{tao2023tma}
\BIBentryALTinterwordspacing
------, ``Enhance security of time-modulated array-enabled directional modulation by introducing symbol ambiguity,'' \emph{arXiv:2310.09922}, 2023. [Online]. Available: \url{https://arxiv.org/abs/2310.09922}
\BIBentrySTDinterwordspacing

\bibitem{oja2000ica}
A.~Hyv{\"a}rinen and E.~Oja, ``Independent component analysis: {Algorithms} and applications,'' \emph{Neural Networks}, vol.~13, no.~4, pp. 411--430, Mar. 2000.

\bibitem{hyvarinen1999gaussian}
A.~Hyvarinen, ``Fast ica for noisy data using gaussian moments,'' in \emph{Proc. IEEE Int. Symp. Circuits Syst.}, vol.~5.\hskip 1em plus 0.5em minus 0.4em\relax IEEE, 1999, pp. 57--61.

\bibitem{Oja2001ICAbook}
A.~Hyv{\"a}rinen, J.~Karhunen, and E.~Oja, \emph{Independent Component Analysis}.\hskip 1em plus 0.5em minus 0.4em\relax Wiley-Interscience, 2001.

\bibitem{hyvarinen1998independent}
A.~Hyv{\"a}rinen and E.~Oja, ``Independent component analysis by general nonlinear hebbian-like learning rules,'' \emph{signal processing}, vol.~64, no.~3, pp. 301--313, 1998.

\bibitem{anton2013elementary}
H.~Anton and C.~Rorres, \emph{Elementary linear algebra: {Aapplications} version}.\hskip 1em plus 0.5em minus 0.4em\relax John Wiley \& Sons, 2013.

\bibitem{johnson1998blind}
R.~Johnson, P.~Schniter, T.~J. Endres, J.~D. Behm, D.~R. Brown, and R.~A. Casas, ``Blind equalization using the constant modulus criterion: A review,'' \emph{Proceedings of the IEEE}, vol.~86, no.~10, pp. 1927--1950, 1998.

\end{thebibliography}
\bibliographystyle{IEEEtran}

\end{document}